\documentclass[compsoc,conference,a4paper,10pt,times]{IEEEtran}
\IEEEoverridecommandlockouts
\usepackage{cite}
\usepackage{amsmath,amssymb,amsfonts}
\usepackage{amsthm} 
\usepackage{graphicx}
\usepackage{textcomp}
\usepackage{bmpsize}
\usepackage{xcolor}
\usepackage{lipsum}
\usepackage{preamble}
\usepackage[colorlinks=true,urlcolor=black]{hyperref}
\def\BibTeX{{\rm B\kern-.05em{\sc i\kern-.025em b}\kern-.08em
    T\kern-.1667em\lower.7ex\hbox{E}\kern-.125emX}}

\begin{document}

\title{Commitment Attacks on Ethereum’s Reward Mechanism}

\makeatletter
\newcommand{\linebreakand}{%
  \end{@IEEEauthorhalign}
  \hfill\mbox{}\par
  \mbox{}\hfill\begin{@IEEEauthorhalign}
}
\makeatother

\author{\IEEEauthorblockN{Roozbeh Sarenche}
\IEEEauthorblockA{\textit{COSIC, KU Leuven} \\
roozbeh.sarenche@esat.kuleuven.be}
\and
\IEEEauthorblockN{Ertem Nusret Tas\thanks{RS and ENT contributed equally and are listed alphabetically.}}
\IEEEauthorblockA{\textit{Stanford University} \\
nusret@stanford.edu}
\and
\IEEEauthorblockN{Barnabé Monnot}
\IEEEauthorblockA{\textit{Ethereum Foundation Research} \\
barnabe.monnot@ethereum.org}
\linebreakand
\IEEEauthorblockN{Caspar Schwarz-Schilling}
\IEEEauthorblockA{\textit{Ethereum Foundation Research} \\
caspar@ethereum.org}
\and
\IEEEauthorblockN{Bart Preneel}
\IEEEauthorblockA{\textit{COSIC, KU Leuven} \\
bart.preneel@esat.kuleuven.be}
}

\maketitle

\begin{abstract}
Validators in permissionless, large-scale blockchains, such as Ethereum, are typically payoff-maximizing, \emph{rational} actors. Ethereum relies on in-protocol incentives, like rewards for correct and timely votes, to induce \emph{honest} behavior and secure the blockchain. However, external incentives, such as the block proposer's opportunity to capture maximal extractable value (MEV), may tempt validators to deviate from honest protocol participation.

We show a series of commitment attacks on LMD GHOST, a core part of Ethereum's consensus mechanism. We demonstrate how a single adversarial block proposer can orchestrate long-range chain reorganizations by manipulating Ethereum's reward system for timely votes. These attacks disrupt the intended balance of power between proposers and voters: by leveraging credible threats, the adversarial proposer can coerce voters from previous slots into supporting blocks that conflict with the honest chain, enabling a chain reorganization.

In response, we introduce a novel reward mechanism that restores the voters' role as a check against proposer power. Our proposed mitigation is fairer and more decentralized – not only in the context of these attacks – but also practical for implementation in Ethereum.
\end{abstract}


\section{Introduction}
\label{sec:introduction}

\subsection{Ethereum Consensus Security}
\label{sec:intro-eth-consensus-security}


Blockchains have emerged as one of the most prominent applications of distributed consensus, a field over $40$ years old.
Whereas earlier blockchains (\eg, Bitcoin) used proof-of-work (PoW) based protocols, the last few years witnessed a proliferation of proof-of-stake (PoS) protocols, where financial stake rather than compute power determines the protocol participants.
The largest proof-of-stake (PoS) blockchain by market capitalization is Ethereum, secured by a PoS Byzantine fault tolerant (BFT) consensus protocol (Gasper~\cite{gasper}).
Ethereum's consensus is run by a set of \emph{validators}, which output \emph{two} ledgers: the \emph{available ledger} and the \emph{finalized ledger}.
Requirements for these ledgers were formally captured by the \emph{ebb-and-flow} property~\cite{ebbandflow}: (i) The \emph{available ledger} must remain secure (\ie, safe and live) under a synchronous network despite unexpected and large drops in the number of active validators (this is also known as the sleepy network model~\cite{sleepy}, or dynamic availability~\cite{genesis}).
(ii) The \emph{finalized ledger} must satisfy \emph{accountable safety}\footnote{Namely, if safety is ever violated, the honest validators can identify the adversarial parties responsible for the violation (\cf Appendix~\ref{sec:slashing}).} under a partially synchronous network and must stay live if sufficiently many validators are active.
Finally, (iii) the finalized ledger must always be a \emph{prefix} of the available ledger.
To achieve the ebb-and-flow property, Ethereum employs two sub-protocols: LMD GHOST (Latest Message Driven
Greedy Heaviest Observed Sub-Tree) that outputs a \emph{canonical chain} as the available ledger, and Casper FFG (the Friendly Finality Gadget~\cite{casperffg}) that checkpoints blocks on this canonical chain.
These checkpointed blocks constitute the finalized ledger.

Ethereum's consensus has been subjected to much scrutiny, as evidenced by attacks on the security of both LMD GHOST~\cite{ebbandflow,3attacks,ethresearch-balancing-attack,ethresearch-balancing-attack2,twoattacks} and Casper FFG~\cite{ffg-attack-1,ffg-attack-2} as well as various protocol patches proposed as response (\cf~\cite{mitigations} for a comprehensive history).
Without the safety of the canonical chain (available ledger), Casper FFG would fail to checkpoint blocks, causing the finalized ledger to stall indefinitely~\cite{ethresearch-bouncing-attack}.
Similarly, without the liveness of the canonical chain, Casper FFG cannot ensure the liveness of the finalized ledger.
Thus, the security of LMD GHOST that outputs the canonical chain is crucial for the security of both ledgers.

Many Ethereum users today act on time-sensitive transactions as soon as they enter the canonical chain, without waiting the $12.8$ minutes required for their finalization by Casper FFG, or any other notion of confirmation by LMD GHOST.
This makes them susceptible to instability on the chain tip.
Making matters worse, validators are incentivized to reorganize (\emph{reorg}) the blocks at the end of the canonical chain by proposing new blocks on top of older blocks rather than the chain tip.
The main reason for this is the magnitude of maximal extractable value (MEV)~\cite{flashboys-website}, \ie, the payoff of a block proposer due to its control over the inclusion, exclusion, and ordering of transactions (\eg, the proposer can frontrun arbitrage transactions in trades, backrun liquidation events to buy the liquidated assets)~\cite{flashboys-2}.
Once the blocks at the chain tip are replaced (\ie, \emph{reorged}), the new blocks can re-include the transactions of the reorged blocks, and capture the associated MEV.
Therefore, besides security, preventing reorgs emerged as a goal of Ethereum consensus, leading to new reorg-resilient proposals to replace LMD GHOST (\eg, Goldfish~\cite{d2022goldfish}, RLMD-GHOST~\cite{rlmd-ghost}).

\subsection{Incentives in Ethereum}
\label{sec:intro-incentives}

LMD GHOST proceeds in slots of $12$ seconds, each with a unique block proposer (hereafter called the \emph{leader}) that proposes a new block extending the tip of the canonical chain.
Each slot is also assigned a committee of validators called \emph{attestors} that vote for the block at the tip of the canonical chain in their views.
To deter future leaders from reorging blocks, these attestors must be incentivized to vote for the proposal at the canonical chain's tip before the next slot starts; so that the proposed block gathers enough support against any future competitor.
Hence, incentivizing the timeliness of these votes is crucial for mitigating reorg attempts and achieving a notion of reorg-resilience in the presence of payoff-maximizing (\ie, \emph{rational}) validators, which might otherwise find it profitable to delay their votes or vote for old blocks.
For this purpose, Ethereum uses a so-called \emph{head vote reward mechanism}: a vote sent for some slot $t$ is rewarded on the canonical chain only if the vote is included in a block $B$ proposed at slot $t+1$ (\ie, \emph{timely}) and if the vote is for the last block from slots $[0,t]$ on the canonical chain (\ie, \emph{correct}).
Ethereum also implements an \emph{inclusion reward mechanism}: a slot $t+1$ leader is rewarded for including the (timely) slot $t$ votes in its block.
Opposing these rewards, future leaders are motivated by MEV to explore ways of deterring previous slots' attestors from sending timely votes.
The tension between the opposing incentives has so far resisted resolution in the form of an attack or a security proof considering rational validators.

\subsection{Commitment Attacks on LMD GHOST}
\label{sec:intro-incentive-attack}

In this work, we present a series of attacks on the reorg-resilience and security of the latest version of LMD GHOST~\cite{mitigations}.
Their scope ranges from reorging a single block to violating the security of LMD GHOST (\ie, the available ledger), which in turn causes a liveness violation for Casper FFG's finalized ledger.
The attacks exploit the head vote reward mechanism in the presence of rational attestors.
Recall that the attestor sending a slot~$t$ head vote is rewarded only if the vote is included in a block of slot $t+1$.
The core idea behind our attacks is that an adversarial leader can abuse its power over the reward of the previous slot's attestors to incentivize them to vote for blocks favorable to the adversary.

The attacks can be conducted by even a single adversarial validator who owns the minimum amount of stake. 
However, the greater the number of validators controlled by the adversary, the more frequently the attacks can be executed.
The adversarial strategy involves only \emph{committing} to a credible threat, which the attestors must be aware of to react rationally (\cf Section~\ref{sec:credibility} for details).
In today's Ethereum, over 90\% of the validators run an out-of-protocol software, `mev boost', in which block proposers outsource the building of their block content to specialized third parties~\cite{mev-boost-stats,mev-boost-github}. 
It is easily imaginable to extend this widely distributed software program to also facilitate the communication between the adversary and attestors in our attacks.
Then, as long as the majority of the attestors are rational (and the remaining ones can be correct), the attacks are guaranteed to be successful.


\smallskip
\noindent
\subsubsection{ Simple Attack (Section~\ref{sec:solo-validators-simple}, Figs.~\ref{fig:simple_game_unsuccessful} and~\ref{fig:simple_game_successful})}
Suppose the leader of some slot $t+1$ is adversarial and aims
to reorg the block $B$ proposed at slot $t$. 
Before slot $t$ starts, it informs the slot $t$ attestors that it is the leader of slot $t+1$,
and asks them to 
vote for $B$'s parent block instead of $B$ itself. 
The adversary commits to including in its block only the votes that \emph{comply} with its instruction.
The slot $t$ attestors are incentivized to comply as otherwise, they do not receive the head vote reward\footnote{Head vote reward constitutes $1/4$ of the total attestation reward~\cite{rewards}, which could be the difference between making a profit vs. loss while validating new blocks.}.  
We show 
a Nash equilibrium in which the attack is successful and the adversary receives the attestation inclusion reward and an additional MEV boost, given both non-colluding (called \emph{solo}) attestors and staking pools of colluding attestors. 

\subsubsection{Strong Simple Attack (Section~\ref{sec:solo-validators-strong-simple})}
Since complying with the adversary is a weakly dominant strategy for the attestors in the simple attack, there exist other equilibria, where a coalition of players can refuse to comply with the adversary to foil the attack.
Hence, in Section~\ref{sec:solo-validators-strong-simple}, we introduce a strong simple game, where an adversary controlling multiple leaders can incentivize the attestors to follow the adversarial strategy regardless of the actions of other attestors, and there is no equilibrium where the adversary fails.
\begin{figure*}[t]
    \centering
    \begin{subfigure}[!hb]{0.49\textwidth}
        \centering
        \includegraphics[height=1.1in]{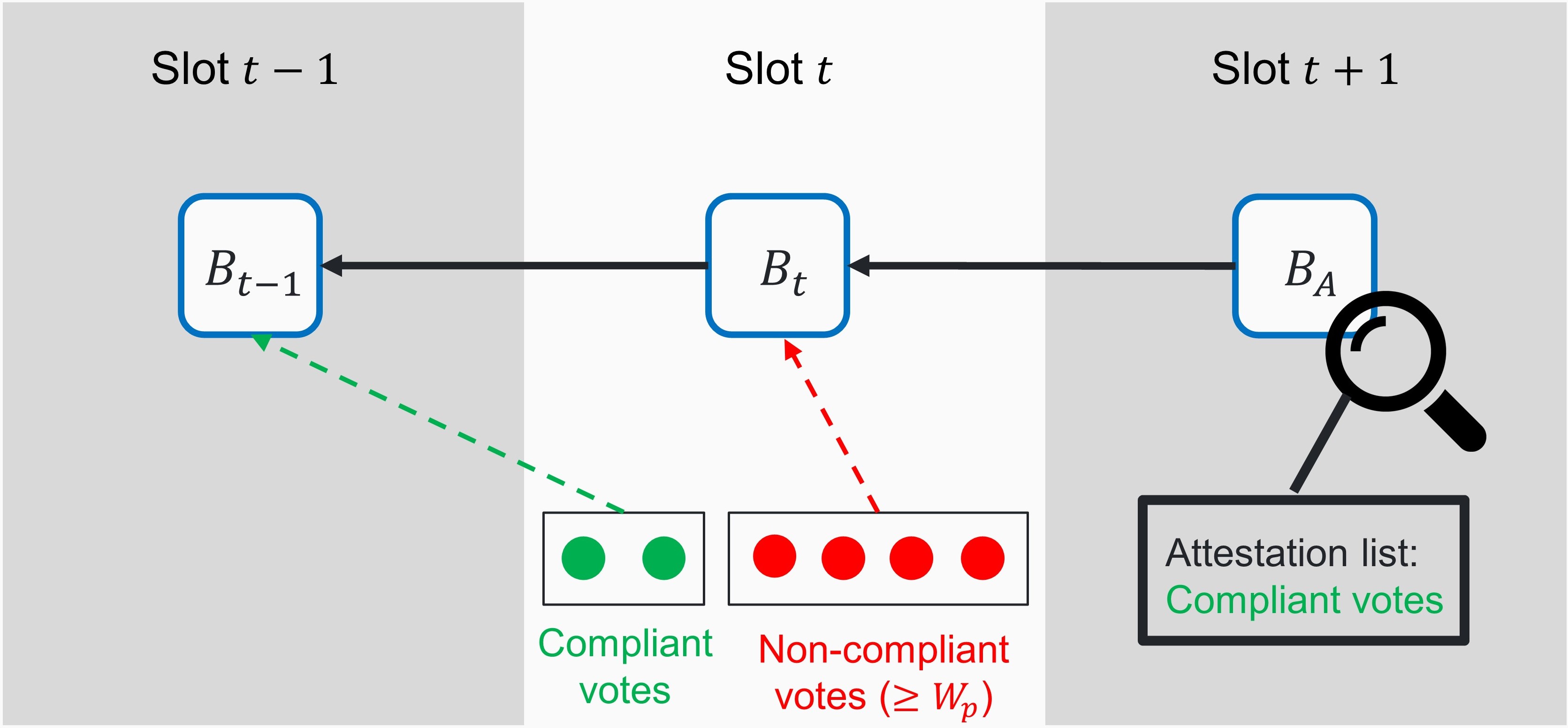}
        \caption{Unsuccessful attack}
        \label{fig:simple_game_unsuccessful}
    \end{subfigure}%
    ~ 
    \begin{subfigure}[!hb]{0.5\textwidth}
        \centering
        \includegraphics[height=1.1in]{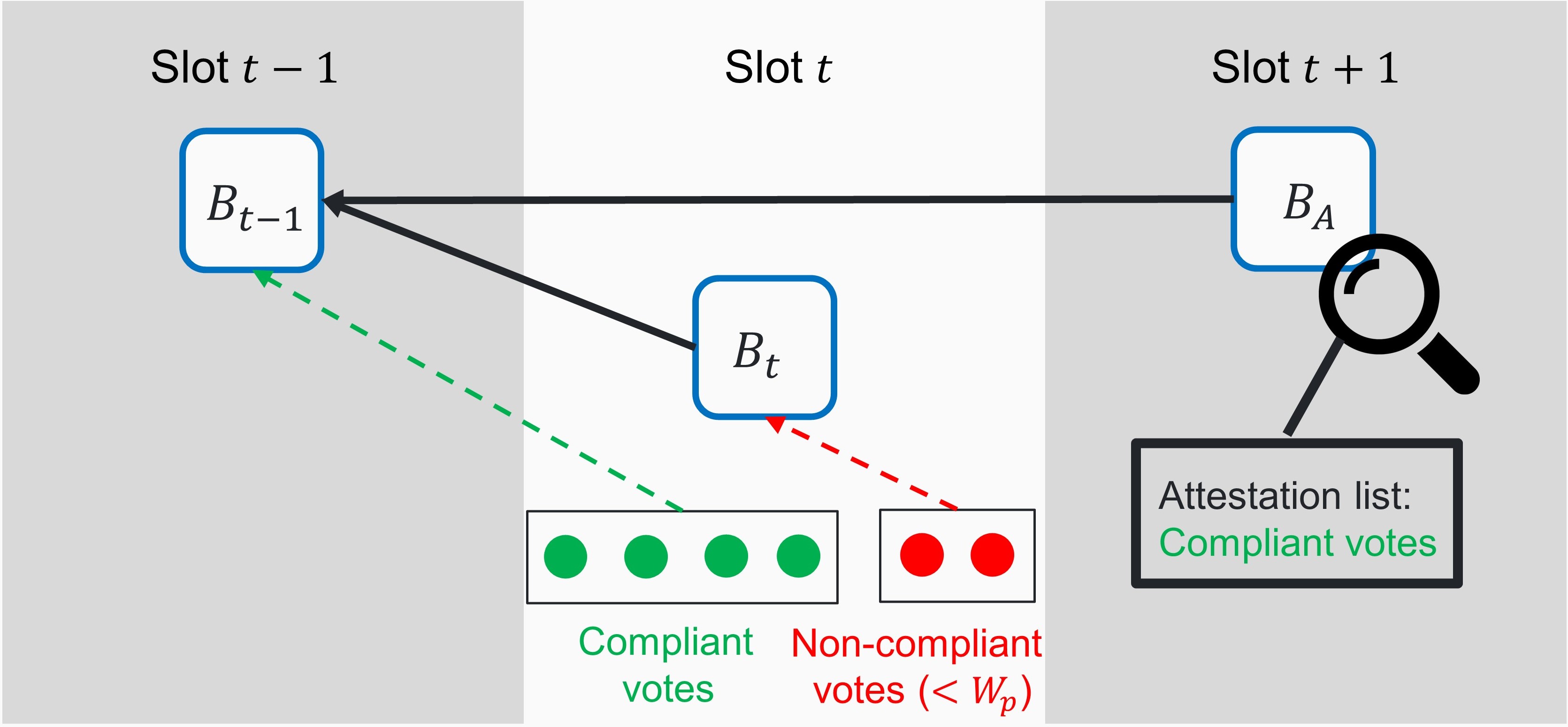}
        \caption{Successful attack}
        \label{fig:simple_game_successful}
    \end{subfigure}
    \caption{The simple game in the presence of non-colluding (solo) validators. As the adversary is committed to excluding the non-compliant votes, even when the attack is unsuccessful, non-compliant votes are not rewarded. $W_p$ denotes the \emph{proposer boost}, equal to 40\% of the slot committee size, that puts a temporary weight on new proposals.}
\end{figure*}

\smallskip
\noindent
\subsubsection{ Extended Attack (Section~\ref{sec:solo-validators-extended}, Fig.~\ref{fig:extended_attack_intro})} 
In the extended attack, the adversary reorgs a sequence of consecutive blocks 
with \emph{empty} (`transaction-less') blocks.
To give insight into the attack, we describe a sketch, where some slot $t$ is normalized to be slot $0$, and the adversary is the leader of slot $3$.
It aims to reorg the two consecutive blocks $B_{-1}$ and $B_{0}$ from slots $-1$ and $0$
with a fork of empty blocks $B_1$ and $B_2$ proposed on top of block $B_{-2}$.
\begin{figure}[!b]
    \centering
    \centering
    \includegraphics[height=1.5in]{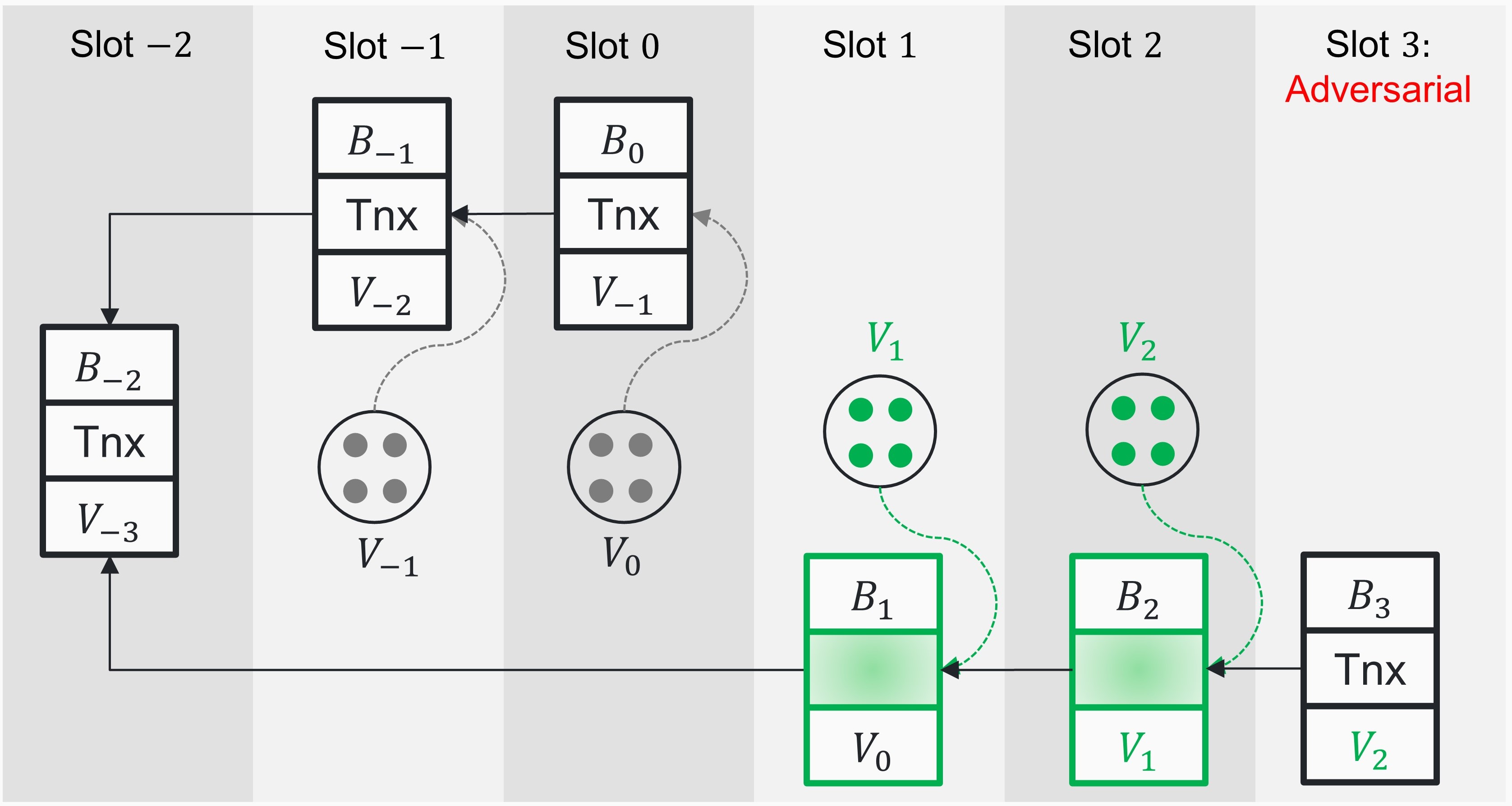}
    \caption{The extended attack with $p=2$. $B_i$ and $V_i$ represent the block and votes for slot $i$, respectively. The votes and blocks in green are compliant.}
    \label{fig:extended_attack_intro}
\end{figure}
%
Using its block $B_{3}$, the adversary can incentivize not only the slot $2$ attestors but also the leaders and attestors of previous slots to follow its instructions; since the slot $2$ votes affect the confirmation status of the slot $2$ block, 
which in turn influences the rewards of the slot $1$ attestors.

Formally, the adversary devises a set of instructions for the leaders and attestors of slots $1$ and $2$ to propose and vote for a fork of empty blocks extending $B_{-2}$ rather than the tip $B_{0}$ of the canonical chain.
Such blocks and votes that comply with the adversary's instructions are called \emph{compliant}.
The adversary stipulates the slot $1$ leader to propose an empty block $B_{1}$ on top of $B_{-2}$, and the slot $1$ attestors to either vote for $B_{1}$ if it is empty and extends $B_{-2}$, or to vote for $B_{-1}$ otherwise.
Now, if $B_{1}$ is compliant and gathers sufficiently many votes, the adversary demands the slot $2$ leader to propose an empty block $B_{2}$, including the compliant slot $1$ votes, on top of $B_{1}$.
It then asks the slot $2$ attestors to either vote for $B_{2}$ if it is empty, includes the compliant slot $1$ votes and extends $B_{1}$, or to vote for $B_{1}$ otherwise (if $B_{1}$ is compliant).
Finally, the adversary commits to including only the compliant slot $2$ votes in its block $B_{3}$.
Note that no rational attestor is obliged to comply with the instructions \emph{a priori}, but as we show next, they are all \emph{incentivized} to do so.

Due to the adversary's commitment, to receive head vote rewards, the slot $2$ attestors vote for $B_{2}$ only if $B_{2}$ is compliant, \ie, is empty and contains only the compliant slot $1$ votes.
Therefore, to ensure that $B_{2}$ receives enough votes and thus cannot be reorged by $B_{3}$, the slot $2$ leader proposes a compliant block $B_{2}$.
Since a compliant $B_{2}$ includes only the compliant slot $1$ votes, to receive head vote rewards, the slot $1$ attestors vote for $B_{1}$ only if $B_{1}$ is compliant, \ie, is empty and extends $B_{-2}$.
Then, to ensure that $B_{1}$ receives enough votes and thus cannot be reorged by $B_{2}$ or $B_{3}$, the slot $1$ leader proposes a compliant block $B_{1}$.
Consequently, all slot $1$ and $2$ leaders and attestors vote for a chain of empty blocks ($B_{1}$ and $B_{2}$) conflicting with $B_{-1}$ and $B_{0}$, enabling the adversary to reorg $B_{-1}$ and $B_{0}$.
%
As shown in Section~\ref{sec:solo-validators-extended}, the extended attack results in a subgame perfect Nash equilibrium with arbitrarily long reorgs.

\subsubsection{Staking Pools (Section~\ref{sec:staking-pools})}
The simple and extended attacks assume the existence of only solo validators. 
However, in practice, staking pools control a large subset of validators, which can collude to mitigate the attacks, indicating a need to revisit the analysis. 
In Section~\ref{sec:staking-pools}, we illustrate Nash equilibria for both the simple and extended attacks, where the adversary succeeds even in the presence of staking pools of bounded size.
However, the extended attack cannot be considered a practical threat if sufficiently many attestors coordinate 
their strategies. 
To this end, we present a new, selfish mining-inspired attack in Section~\ref{sec:staking-pools}, where all staking pools, regardless of their size,
are incentivized to comply with the adversary under mild assumptions on the fraction of adversarial validators. 

\subsubsection{Attack Quantification (Appendix~\ref{sec:appendix_attack_quantification})} 
Our analysis of both the consensus layer reward and the execution layer rewards shows that upon a successful reorganization of a single block, a single adversarial slot leader can achieve an additional reward of $71.1$ million Gwei ($170$ EUR\footnote{As of 23 October 2024, the price of 1 ETH is approximately 2400 EUR~\cite{eth_price}.}), representing a $56\%$ increase in the block reward. To better illustrate the economic impact behind reorg attacks, consider that if the largest Ethereum staking pool\footnote{As of October 23, 2024, the largest staking pool is LIDO, holding $27.8$ of stake shares\cite{lido}.} manages to reorg one block for each of its blocks, it could collect, on average, an extra reward of almost 234,000 EUR per day.

\subsection{A New Reward Mechanism: \NewMechanism}
\label{sec:introduction-solution}

\begin{figure}[b]
    \centering
    \includegraphics[height=2in]{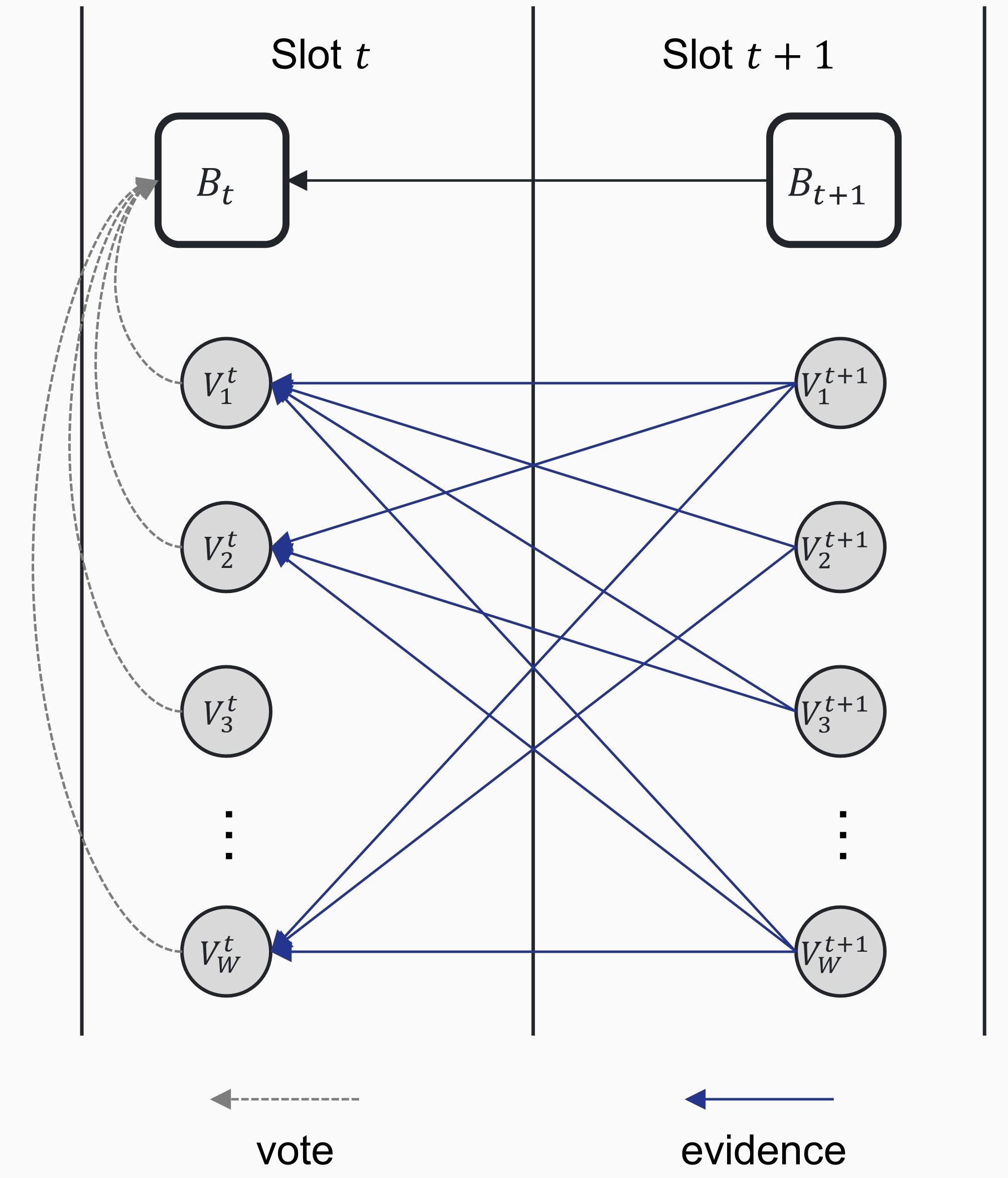}
    \caption{The \newmechanism reward mechanism. $V^t_i$ represents the $i^\text{th}$ vote (by the $i^\text{th}$ attestor) in slot $t$. Vote $V^t_1$ has received signatures from more than half of the slot $t+1$ attestors, \ie, it is timely. Vote $V^t_3$ has not received any signature from slot $t+1$ attestors and is not timely.}
    \label{fig:new-mechanism}
\end{figure}

The attacks above are due to the \emph{monopoly} of a slot leader over which votes get included and subsequently rewarded by the canonical chain.
Hence, the protocol can be reinforced against the leaders' abuse of power by \emph{decentralizing} the reward mechanism~\cite{decentralization}.
Towards this goal, we propose a new vote reward mechanism for LMD GHOST (Fig.~\ref{fig:new-mechanism}, Section~\ref{sec:solution}) called \emph{\newmechanism}, reflecting the directed acyclic graph (DAG) topology of vote connections.
\Newmechanism replace the leader of slot $t+1$ with the slot $t+1$'s committee of attestors when determining whether slot $t$ votes were timely.
It requires each slot $t+1$ attestor to sign the slot $t$ votes in its view.
Then, an attestor is rewarded for its slot $t$ vote if there are sufficiently many slot $t+1$ attestors that have signed the 
vote, \emph{whenever} these signatures are included on-chain\footnote{The inclusion reward mechanism used by Ethereum can be extended to incentivize the inclusion of the slot $t+1$ attestors' signatures on the slot $t$ votes.}.
The adversary must now control a large fraction of the attestor committees to affect the head vote rewards.
In Section~\ref{sec:dag-votes-analysis}, we prove the existence of a subgame perfect Nash equilibrium, in which LMD GHOST equipped with DAG votes remains secure under some mild assumptions.

\Newmechanism not only mitigate the attacks above,
but also 
protect the honest attestors' rewards.
Currently, if a slot leader does not propose any block, the attestors of the previous slot lose their head vote rewards,
which is particularly discouraging for solo validators selected as attestors once per epoch.
The new mechanism ensures rewards even when a minority of the attestors are offline.

\subsubsection{DAG Votes Made Practical}
Considering the number of Ethereum validators ($>1$ million), the DAG votes mechanism imposes a high computational and communication overhead on the Ethereum network. 
To address this issue, in Section~\ref{sec:practical}, we introduce a practical version of the mechanism that is backward-compatible with the Ethereum infrastructure and can be implemented at the cost of a slight increase in block space, computation, and communication resources. 
In the practical version, only a subset of the slot $t+1$ attestors, referred to as aggregators, become responsible for signing the slot $t$ attestations.
Under synchrony, these signatures can be quickly aggregated into a single one, resulting in only a slight increase in the required block space.

In Section~\ref{sec:implementation_DAG}, we present an implementation of a simplified version of our DAG votes mechanism, deployed on a peer-to-peer local blockchain running across multiple Amazon Web Services instances.
Our experiments show that the block space increases by merely 163 bytes due to the DAG votes mechanism, with minimal increase in the average creation and verification time for the aggregated attestations. 
Additionally, in Appendix~\ref{sec:appendix-optimistic}, we analyze the overhead of our scheme if implemented on the Ethereum mainnet (as opposed to a few instances as in our experiments).
Our analysis shows that under the new mechanism, the block space increases by merely 14 kilobytes, which represents a $14\%$ increase in block size. 
We also find that on a consumer laptop, the new mechanism imposes a computational overhead of an additional 0.69 seconds, 1.36 seconds, and 0.4 seconds for an Ethereum aggregator, block proposer, and validator, respectively, compared to the current version of Ethereum.
Regarding the communication overhead of the new mechanism, each aggregator needs to send an extra pair of messages: an aggregated attestation and a signature on the aggregated attestation, resulting in only 192 extra bytes of message. 

\section{Related Work}
\label{sec:related-work}

\subsection{Attacks on Ethereum Consensus}
\label{sec:related-work-existing-attacks}
%

The LMD GHOST protocol as introduced by Buterin \emph{et al.} in~\cite{gasper} is susceptible to the balancing attacks~\cite{ebbandflow,3attacks,ethresearch-balancing-attack,ethresearch-balancing-attack2}, where the adversary causes two equally-sized sets of honest validators to vote for conflicting blocks indefinitely under synchrony by carefully timing the release of a few withheld adversarial votes.
The protocol was subsequently changed to bestow the latest block proposals with a temporary extra weight called the proposer boost, so as to push the validators towards the latest proposal~\cite{mitigations}.
However, proposer-boosting failed to prevent a resurgence of the balancing attack due to a feature of the LMD rule that requires each validator to consider only the first slot $t$ vote it receives from an equivocating validator.
The absence of this rule in turn enabled a so-called avalanche attack that showed the incompatibility of the PoS Sybil resistance mechanism with the GHOST fork-choice rule~\cite{twoattacks}.
Finally, these attacks were addressed with another patch called equivocation discounting, removing equivocating votes for the same slot but different blocks from the validators' views~\cite{mitigations}.
A confirmation rule and the associated security proof for LMD GHOST are provided in~\cite{ghost-confirm}.

In our work, we investigate the effect of the reward mechanism on the protocol's security \emph{in the presence of rational validators}.
Our attacks do not require the adversary to 
control more than a single block proposer or send equivocating blocks/votes, making the bar for the attack much lower than the earlier consensus attacks.
In fact, when the attack is successful, the adversary does not even have to violate the protocol.

\subsection{Analogous Attacks on PoW Blockchains}
The attacks introduced in this paper are similar to those in PoW blockchains, such as selfish mining~\cite{selfish_mining}, undercutting~\cite{carlsten2016instability} and bribery attacks~\cite{Bonneau16}, which aim to orphan non-adversarial blocks. 
In the selfish mining attack~\cite{selfish_mining}, the adversarial miner withholds a newly mined block and releases it later to orphan other miners' blocks. 
Extensive literature~\cite{sapirshtein2017optimal, bar2022werlman, sarenche2023deep, nayak2016stubborn, wang2021blockchain, mirkin2020bdos} exists analyzing the impact of selfish mining on the security of longest-chain protocols (Nakamoto consensus), such as Bitcoin. 
In the undercutting attack~\cite{carlsten2016instability, gong2022towards}, the adversarial miner attempts to orphan the tip of the chain by incentivizing others to mine on its block during a fork race, using transaction fees as a form of bribery.
In bribery attacks~\cite{Bonneau16}, the adversary bribes miners to temporarily control a majority of the hash power and attack the chain.

Note that there is an important difference between the incentives for orphaning (or reorging) blocks in PoW blockchains and those in Ethereum. 
In PoW, the primary reason orphaning blocks is profitable is that it reduces mining difficulty, thereby accelerating the process of mining blocks and collecting rewards~\cite{grunspan2018profitability, sarenche2024time}. 
In contrast, in Ethereum, the main incentive for reorging blocks is to steal MEV and inclusion rewards from other blocks.

\subsection{Timing Games}
\label{sec:related-work-timing-games}
Our model of validators shares some similarities with the honest-but-rational validators of \cite{timing-games,timing-games-2}.
These validators do not strictly violate the protocol rules but delay their blocks to the maximum extent possible to capture more MEV, while ensuring their timely inclusion in the canonical chain.
The advantage of well-connected validators and the resulting higher profit in MEV has serious potential to hurt fairness and decentralization by incentivizing collusions for more connectivity.
Like these works, our model also assumes rational, payoff-maximizing validators, and just like delaying blocks, protocol deviations by our rational validators are not \emph{accountably} identifiable or \emph{slashable} (\ie, punishable financially).
However, unlike honest-but-rational proposers, our validators can explicitly disobey the protocol rules that are not enforceable via slashing, \eg, by refusing to propose on the tip of the canonical LMD GHOST chain.

\subsection{Game Theory and Consensus Protocols}
\label{sec:incentives-consensus}
Behavior of distributed protocols with payoff-maximizing participants has been studied for secret sharing and multi-party computation, that were shown to have Nash equilibria with desirable properties under bounded collusions~\cite{game-theory-distributed-computing}.
Kiayias and Stouka in~\cite{kiayias2021coalition} have introduced the concept of coalition-safe equilibria to investigate the potential increase in participant payoff through protocol deviation in Bitcoin and Fruitchain~\cite{fruitchain}.
In our modelling of validators, we follow the BAR (Byzantine, altruistic, rational) model of protocol participants introduced by~\cite{bar-model}, which we respectively call adversarial, honest and rational.
The same work~\cite{bar-model} provided the first safe and live asynchronous state machine replication protocol in the presence of both adversarial (Byzantine) and rational participants.
For a taxonomy of the incentive mechanisms of major blockchain protocols, we refer the reader to~\cite{incentive-survey-1}.

\section{Preliminaries}
\label{sec:preliminaries}

\subsection{Model}
\label{sec:model}

Given any event $A$, let $NA$ denote the \emph{complement} of the event $A$ such that $A \cap NA = \emptyset$ and $\Pr[A \cup NA] = 1$.
We define notation $[p]$, $p \in \mathbb{N}$, to denote the set $\{1, \ldots, p\}$.

\subsubsection{ State machine replication and blockchains}
In state machine replication consensus protocols, a set of nodes called the \emph{validators} agree on a growing sequence of transactions (or blocks) called the \emph{ledger} (or chain) and denoted by $\ch$.
Each validator can thus obtain the same end state upon executing these transactions in the order determined by $\ch$.
%
%
Each chain starts with a publicly known genesis block $B_0$.
Each block except $B_0$ refers to its \emph{parent} through its hash.
A block $B$ \emph{extends} $B'$, denoted by $B' \preceq B$, if $B = B'$ or $B'$ lies on the path from $B$ to $B_0$.
All valid blocks extend $B_0$.
For each $B$, the path from $B$ to $B_0$ determines a unique chain $\ch$.
The chain held by a validator $\val$ at time $t$ is denoted by $\ch^{\val}_t$.

\subsubsection{ Adversary}
The adversary $\Adv$ is an efficient algorithm that controls a subset of the validators called \emph{adversarial}.
These validators can violate the consensus rules in an arbitrary fashion coordinated by $\Adv$ (Byzantine faults).
The remaining validators are either honest, \ie, follow the prescribed protocol, or \emph{rational}, \ie, take the action providing the largest expected payoff in any specified game. 
In later sections, we will first consider a model with adversarial and \emph{solo validators}, which are rational and do not collude.
Afterwards, we will analyze the interaction between the adversarial validators and 
colluding sets of rational validators.

\subsubsection{ Network}
We consider a synchronous network, where the adversary can control the delivery time of the messages sent by the rational and honest validators up to a known delay upper bound $\Delta$.
We hereafter normalize $\Delta$ to be $1$ and adopt a lock-step communication model, where messages can only be sent at discrete intervals, and a message sent at some time $t$ (\ie, wall clock time $\Delta t$) is delivered by time $t+1$, (\ie, $\Delta t+ \Delta$).
Since our attacks do not require the adversary to delay the honest messages in arbitrary ways, we assume that the views of all validators and the adversary are the same at all `lock-step' times.
A validator is said to \emph{broadcast} a message if its recipients include all validators.

\subsubsection{ Security}
%
An SMR protocol is secure with latency $\Tconfirm$ if:
%
    
    \smallskip
    \noindent
    \textbf{Safety:} For any times $t,t'$ and validators $\val,\val'$, either $\ch_{t}^{\val} \preceq \ch_{t'}^{\val'}$ or vice versa. 
    For any $t' \geq t$, $\ch^{\val}_{t} \preceq \ch^{\val}_{t'}$.
    
    \smallskip
    \noindent\textbf{Liveness:} If a transaction $\tx$ is input to an honest validator at some time $t$, then there exists a finite time $t^* > t$ such that $\tx \in \ch_{t'}^{\val}$ for any $t' \geq t$ and any $\val$.
%

\subsection{The LMD GHOST Protocol}
\label{sec:lmd-ghost-background}

The LMD GHOST 
protocol proceeds in \emph{epochs} with $32$ \emph{slots}, each with duration $3\Delta = 3$. 
Before an epoch $e$, a randomness beacon called RANDAO divides the validator set of epoch $e$ into $32$ disjoint \emph{committees} of size $W$ assigned to $32$ unique slots within the epoch~\cite{shuffling}.
It also selects a validator from each committee as the unique \emph{leader} for its slot.
A slot with an adversarial (rational, honest) leader is called an \emph{adversarial} (rational, honest) slot.
The adversary knows which slots are adversarial before the epoch starts\footnote{In RANDAO, slot leaders know their slots before the epoch.}. 

Validators use the LMD GHOST fork-choice rule to identify a \emph{canonical} (beacon) chain of blocks (later, Casper FFG finalizes the blocks on this canonical chain and thus outputs the finalized chain).
As we assume that all validators have the same view of the protocol, we hereafter omit the mention of views.
At the beginning of a slot $t$, the leader $L$ signs and proposes a new block $B = \langle e, t, \txs \rangle_L$ with transactions $\txs$, that extends the canonical chain tip.
After $\Delta$ time, each slot $t$ attestor $\val$ signs and broadcasts a vote $\langle e, t, B \rangle_\val$ for the block at the tip of the canonical chain.
Within $\Delta$ time, these votes are received and aggregated by a special subset of attestors, 
who in turn broadcast the aggregated votes to be included by the next slot's block~\cite{aggregators}. 
A slot $t$ vote is \emph{valid} only if it is for a block with slot $t' \le t$, and the vote was sent by a slot $t$ attestor.

The votes sent within LMD GHOST are in fact called head votes~\cite{rewards} to distinguish them from the Casper FFG votes. 
Since we exclusively focus on LMD GHOST in this work, we will refer to head votes as simply votes unless specified otherwise.

When invoked at slot $t$, the fork-choice rule obtains the canonical chain iteratively starting at $B_0$.
At each block $B$, it inspects the subtrees rooted at $B$'s children.
It then selects the child block, whose tree has the largest weight, ties broken by the adversary.
The weight is equal to the number of unique votes by the validators for the blocks in the tree, plus a proposal boost $\Wp$ (currently $\Wp/W = 40\%$) if the tree contains a block proposed for the current slot $t$\footnote{The original GHOST rule counted blocks rather than votes~\cite{ghost}}.
If there are multiple valid votes by the same validator within a tree, the rule ignores all but one of the votes for the latest slot, thus the name `latest message driven (LMD)'.
In future sections, we will consider slots from the same epoch, with distinct attestor committees, unless stated otherwise.

\subsubsection{ Reorgs}
Suppose a block $B$ proposed at some time $t'$ is not in the canonical chain for the first time at some $t \geq t'$.
Then, $B$ is said to be \emph{reorged} at time $t$.
Reorg of a few honestly proposed blocks need not imply a liveness violation; however, reorgs of arbitrarily long block sequences would be a violation of LMD GHOST's safety and liveness
(\cf~Appendix~\ref{sec:appendix-reorgs}).

\subsection{The Ethereum Reward Mechanism}
\label{sec:reward_mechanism_background}

We review two main sources of the validator's consensus layer reward: the inclusion (proposing) and the attestation rewards. 
The inclusion reward is the reward that a slot leader receives for including \emph{attestations} in its proposed block. 
The attestation reward is the reward that a validator receives for sending attestations.
Each attestation consists of $3$ different votes: the source checkpoint, the target checkpoint, and the head vote (\cf Section~\ref{sec:lmd-ghost-background}). 
An attestation is eligible to receive the attestation reward if it is included in the canonical chain and is \emph{correct} and \emph{timely}~\cite{rewards}. 
Here, we only elaborate on the correctness and timeliness of the head votes (hereafter called just the votes), moving the discussion of these attributes for other votes to Appendix~\ref{sec:reward_mechanism_background_appendix}.

A slot $s$ vote for a block is correct if that block is the last block within the canonical chain from slots $s' \in [0,s]$.
If the tip of the canonical chain in the view of an attestor is extended by a block $B_t$ at slot $t$, the attestor should vote for block $B_t$. 
However, if the slot block is missing, the attestor can vote for the previous block in the canonical chain, and the vote is still correct. 
Since Ethereum's Altair upgrade~\cite{altair},
a slot $t$ vote must be included in the \emph{subsequent} block of slot $t+1$ to be considered timely.
A complete overview of the Ethereum reward mechanism is presented in Appendix~\ref{sec:reward_mechanism_background_appendix}.

\section{The Simple Attack in the Presence of Solo Validators}
\label{sec:solo-validators-simple}
The attacks 
exploit Ethereum's reward mechanism to incentivize validators to act in favor of the adversary. 
%
%
%
%
%
Each attack can be viewed as a `game' induced by the adversary that commits to a fixed set of actions and communicates this to the validators.
The players are the rational validators that best respond to the actions committed by the adversary.
Depending on the game, the players can include the slot leaders, attestors or both.
The possible actions available to the players are as follows:
\begin{itemize}[leftmargin=*]
    \item A slot $t$ leader can build a \emph{single} new block for slot $t$, denoted by $B_{t}$, extending any block in its view at any time. 
    \item A slot $t$ attestor can send a \emph{single} slot $t$ vote for any block in its view at any time.  
\end{itemize}
We assume that a rational validator never proposes two different blocks for the same slot or sends two or more votes with the same slot number for distinct blocks.
This is because these violations are provably detectable and cause the slashing (burning) of the validator's stake.
Although the honest actions for the slot $t$ leader and attestors are to respectively send their block and votes at times $3t$ and $3t+1$, the rational validators can delay their messages or withhold them if that gives them a higher payoff.

The players' payoffs come from the inclusion and attestation rewards (\cf, Section~\ref{sec:reward_mechanism_background}). 
We denote by $r$ the reward an attestor receives for submitting a single correct and timely vote and by $R$ the reward a slot leader receives for including a single correct and timely vote in its block.

\subsection{The Simple Game}
\label{sec:simple-game}
The simple game with solo validators (Fig.~\ref{fig:simple_game_unsuccessful}) starts at the beginning of some slot $t$ at time $3t$ such that the leader of slot $t+1$ is adversarial.
The players are the slot $t$ attestors that respond to the following strategy by the adversary:
\begin{definition}[Adversary's Strategy for the Simple Game]
\label{def:adversary-action-simple}
    Let $B_{t-1}$ denote the block at the tip of the canonical chain at the end of slot $t - 1$.
    At the start of slot $t$, the adversary $\mathcal{A}$ 
    %
    \begin{enumerate}[leftmargin=*]
        \item Generates a proof asserting its leadership in slot $t+1$.
        \item Sets the instructions for the validators (called the game rule) 
        as follows: The attestors of slot $t$ should vote for the block $B_{t-1}$
        by the end of slot $t$.
        \item Commits to the following action: Votes of the attestors that do not vote for $B_{t-1}$ will not be included in the adversary's block $B_A$ in slot $t+1$.
        \item Broadcasts the leadership proof, the game rule and its commitment
        to all slot $t$ attestors.
    \end{enumerate}
    %
    At the start of slot $t+1$, adversary $\mathcal{A}$
    \begin{enumerate}[leftmargin=*]
        \item Proposes a block $B_A$ that includes only the slot $t$ votes for the block $B_{t-1}$.
        \item If $\Wp$ or more slot $t$ attestors have voted for $B_{t}$ (block of slot $t$) by time $3t+3$, it proposes its block $B_A$ on top of block $B_{t}$\footnote{In this case, the compliant votes included in $B_A$ only receive the source and target checkpoint rewards and do not receive the head vote reward. This is because the votes are for $B_{t-1}$ and do not satisfy the correctness property, as they are not for $B_A$'s parent $B_t$.}.
        In this scenario (Fig.~\ref{fig:simple_game_unsuccessful}), the attack \textbf{fails}.
        \item If less than $\Wp$ attestors have voted for $B_{t}$ by time $3t+3$, it proposes its block $B_A$ on top of $B_{t-1}$.
        In this scenario (Fig.~\ref{fig:simple_game_successful}), the attack \textbf{succeeds}.
    \end{enumerate}
    %
\end{definition}

The simple game ends, and the payoffs are realized at the end of slot $t$, at time $3t+3$.
Although the rewards for the votes included in $B_A$ are actually realized upon $B_A$'s finalization, we assume that the behavior of the adversarial and rational validators after the game ends does not affect $B_A$'s finalization, which happens within a short period ($\sim 12$ minutes for Casper FFG).

\subsection{The Simple Game Analysis}
\label{sec:simple-game-analysis}
%
The slot $t$ attestors who follow (do not follow) the game rule above are called compliant (non-compliant).
\begin{theorem}
\label{thm:simple-game-theorem}
There exists a Nash equilibrium of the simple game in the presence of solo validators, where the attack succeeds.
\end{theorem}
\begin{proof}[Proof of Theorem~\ref{thm:simple-game-theorem}]
The pay-off matrix of a solo slot $t$ attestor $\val$ is presented by Table~\ref{table:pay-off_matrix-simple_game}.
If over $W_p$ slot $t$ attestors vote for $B_t$, then $B_A$ extends $B_t$ and does not include any of the votes by the non-compliant attestors ($\texttt{Fail}$-$NC_\val$ corner).
In this case, none of the votes by the compliant attestors included in $B_A$ can be correct either, as they would be for the parent of $B_t$ rather than $B_t$ itself ($\texttt{Fail}$-$C_\val$ corner). 
Thus, the payoff is $0$ for both actions $C_\val$ and $NC_\val$ when attack fails.
In contrast, if less than $W_p$ slot $t$ attestors vote for $B_t$, $B_A$ extends $B_{t-1}$ and stays in the canonical chain after slot $t$.
Therefore, the compliant attestors, whose votes within $B_A$ would be correct and timely, receive a payoff of $r$ ($\texttt{Succeed}$-$C_\val$ corner), whereas no votes by the non-compliant attestors are included in $B_A$, making their payoff $0$ ($\texttt{Succeed}$-$NC_\val$ corner).
Hence, voting for $B_{t-1}$ weakly dominates voting for $B_{t}$ or other blocks, and the case where solo slot $t$ attestors follow the game rule is a Nash equilibrium.
\begin{table}
    \centering
    \caption{Pay-off matrix of a solo slot $t$ attestor in the simple game. Let $C_\val$ denote the event that the validator $\val$ votes for $B_{t-1}$ (\ie, $\val$ is compliant). 
   Let $\texttt{Succeed}$ denote the event that fewer than $\Wp$ slot $t$ attestors vote for $B_t$ (the attack succeeds), and $\texttt{Fail}$ denote the event that greater than or equal to $\Wp$ slot $t$ attestors vote for $B_t$ (the attack fails).}
    \begin{tabular}{|c|c|c|}
        \hline
          & $C_\val$ & $NC_\val$ \\
         \hline
         $\texttt{Succeed}$ & $r$ & $0$ \\
         \hline
         $\texttt{Fail}$ & $0$ & $0$\\
         \hline
    \end{tabular}
    \label{table:pay-off_matrix-simple_game}
\end{table}
%
%
\end{proof}

\begin{note}
\label{note:4.1}
There are many Nash equilibria, where the adversary $\mathcal{A}$ is not successful. 
Suppose over $\Wp+1$ slot $t$ attestors vote for $B_t$ (the attack fails). 
Then, $\mathcal{A}$ cannot reorg $B_t$ with $B_A$, and no vote for a block other than $B_t$ can be both correct and timely.
As $\mathcal{A}$ also commits to excluding votes for $B_t$, no solo slot $t$ attestor receives a positive payoff for any action.
Therefore, 
voting for $B_t$ weakly dominates any other action for the slot $t$ attestors, making this a Nash equilibrium.
\end{note}


\begin{note}
\label{note:4.3}
Even under an honest minority assumption (\ie, $<\Wp$ honest attestors per slot), LMD GHOST is still susceptible to the simple attack, \ie, there is a Nash equilibrium where the adversary succeeds; since the number of honest validators would not be sufficient to overcome the proposer boost exploited by the adversary's block. 
\end{note}

In Appendix~\ref{sec:appendix_attack_quantification}, we present a quantification of the reorganization attack to estimate the average reward an adversary gains from a successful simple game attack, as well as the cost incurred in the event of an unsuccessful attempt.



\subsection{Credibility}
\label{sec:credibility}
So far, we have considered an adversary that commits to carrying out a certain action.
Although the adversary was not modeled as a rational player, a major motivation for such attacks is financial profit: by reorging $B_{t}$, it can include a more valuable set of transactions in its block $B_A$, resulting in a higher profit via inclusion rewards and MEV capture.
However, since a slot leader profits from including \emph{all} correct votes due to the inclusion reward, if the adversary were financially motivated, its threat about excluding the non-compliant votes would not be considered \emph{credible} by the rational players, unless the adversary can irrevocably commit to following upon its threat.
One way the adversary can enforce credibility is by depositing collateral, larger than the maximum possible inclusion reward, in a public contract that burns the funds if the committed strategy is violated. 
Although this requires the adversary to deposit new collateral, this initial cost can be bypassed in Ethereum using restaking protocols.

The restaking protocols, such as Eigenlayer~\cite{eigenlayer}, are mechanisms that enable Ethereum stakers to leverage their staked Ether (ETH) to secure multiple services and applications simultaneously. 
They allow users to \emph{restake} their existing ETH stake in various new protocols or services without having to lock up additional funds.
Although restaking protocols offer enormous opportunities to Ethereum validators, they can also eliminate the initial cost of commitment attacks, a major concern for cryptoeconomic security.
Indeed, using restaking protocols, an adversarial validator can reuse its ETH stake as collateral in the public contract instead of putting up new deposit.
Moreover, the adversary can set the withdrawal address for its inclusion rewards to be a contract that complies with the restaking protocol, and then specify the inclusion of non-compliant votes as an additional slashing condition for burning this reward.

\section{The Strong Simple Attack in the Presence of Solo Validators}
\label{sec:solo-validators-strong-simple}

The strong simple attack targets a slot $t$ at time $3t$ in epoch $e$, such that the adversary is the leader of slot $t+1$ and is presumed to be the slot leader in at least one slot belonging to the upcoming epochs $\{e + 2, e + 3, \ldots\}$\footnote{If the adversary's stake share is greater than $\epsilon$, there should exist a slot after the start of epoch $e + 2$ whose leader is the adversary.}. 
It is similar to the simple game, with the difference that the adversary is the leader of an additional slot $t' + 1$, where $t' \in \{e + 2, e + 3, \ldots\}$. 
Slot $t'$ should belong to an epoch $e' \ge e + 2$ to ensure that the slot $t'$ attestors are unknown at the time of the attack execution at slot $t$ in epoch $e$. We denote the adversarial block of slots $t+1$ by $B_A$ and refer to the adversarial block of slot $t' + 1$, denoted by $B'_A$ as the supporting block.

\begin{definition}[Adversary's Strategy for the Strong Simple Game]
\label{def:adversary-action}
    Let $t'+1$ denote the first adversarial slot after the start of epoch $e+2$.
    At the start of slot $t$ in epoch $e$, the adversary $\mathcal{A}$
    \begin{enumerate}[leftmargin=*]
        \item Generates a proof asserting ownership over the private keys of adversarial validators, including the adversarial leader of slot $t + 1$.
        \item Sets the instructions for the validators (called the game rule) 
        as follows: The attestors of slot $t$ should vote for the block $B_{t-1}$
        by the end of slot $t$.
        \item Commits to the following action: The vote of a non-compliant attestor (one that does not vote for $B_{t - 1}$) will not be included in the adversary's block $B_A$ in slot $t + 1$. Furthermore, if a non-compliant slot $t$ attestor is among the attestors of slot $t'$ in epoch $\ge e+2$, its vote for slot $t'$ will also not be included in the adversary's block $B'_A$ in slot $t' + 1$.
        \item Broadcasts the ownership proof, the game rule, and its commitment
        to all slot $t$ attestors.
    \end{enumerate}
    %
    
    At the start of slot $t+1$, adversary $\mathcal{A}$ follows the same action as in simple game (Definition~\ref{def:adversary-action-simple}). 
    At the start of slot $t'+1$, adversary $\mathcal{A}$ 
    %
    \begin{itemize}[leftmargin=*]
        \item Includes only the votes of those slot $t'$ attestors who were either not among the slot $t$ attestors or were among the compliant slot $t$ attestors.
    \end{itemize}
    %
\end{definition}

The strong simple game ends at the end of slot $t$, at time $3t+3$. However, the payoffs are realized at the end of the supporting slot $t'$, at time $3t'+3$.

\subsection{The Strong Simple Game Analysis}
\label{sec:strong-simple-game-analysis}
%

\begin{theorem}
\label{thm:strong-simple-game-theorem}
There exists a strong Nash equilibrium of the strong simple game in the presence of solo validators, where the adversary succeeds.
\end{theorem}
\begin{proof}[Proof of Theorem~\ref{thm:strong-simple-game-theorem}]
In every epoch of Ethereum consisting of $32$ slots, each validator is assigned to be the attestor of exactly one slot. Therefore, a slot $t$ attestor in epoch $e$ can be a slot $t'$ attestor in an epoch $\geq e + 2$ with probability $\frac{1}{32}$.

The expected pay-off matrix of a solo slot $t$ attestor $\val$, resulting from head vote rewards in both slot $t$ and $t'$, is presented in Table~\ref{table:pay-off_matrix-strong_simple_game}.
Vote of a slot $t$ non-compliant attestor (regardless of the attack result) is excluded from block $B_A$. Also, if the non-compliant attestor is among slot $t'$ attestors, the adversary excludes its vote from block $B_A'$. Therefore, the expected reward for a non-compliant attestor over 2 slots $t$ and $t'$ is 0 ($\texttt{Fail}$-$NC_\val$ and $\texttt{Succeed}$-$NC_\val$ corners).

Consider a compliant slot $t$ attestor $\val$. If fewer than $W_p$ slot $t$ attestors vote for $B_t$, the attack succeeds. In this case, block $B_A$ includes the slot $t$ vote of $\val$, which is considered a correct vote. Additionally, if $\val$ is among the slot $t'$ attestors, its slot $t'$ vote is also included in block $B_A'$. Therefore, the expected reward for a compliant slot $t$ attestor over 2 slots $t$ and $t'$ under a successful attack is $r + \frac{r}{32}$ ($\texttt{Succeed}$-$C_\val$ corner).

Assume that $\val$ is a compliant slot $t$ attestor. If $W_p$ or more slot $t$ attestors vote for $B_t$, the attack fails. In this case, the slot $t$ vote of $\val$ is considered a non-correct vote. However, if $\val$ is among the slot $t'$ attestors, its slot $t'$ head vote will still be included in block $B_A'$. Therefore, the expected reward for a compliant slot $t$ attestor over 2 slots $t$ and $t'$ under a failed attack is $\frac{r}{32}$ ($\texttt{Fail}$-$C_\val$ corner).
Hence, voting for $B_{t-1}$ strongly dominates voting for $B_{t}$ or other blocks, and the case where solo validators follow the game rule is a strong Nash equilibrium.
\begin{table}
    \centering
    \caption{Expected payoff matrix of a solo slot $t$ attestor in the strong simple game resulting from head vote rewards in both slot $t$ and $t'$. Let $C_\val$ denote the event $\val$ is compliant. 
     Let $\texttt{Succeed}$ denote the event of a successful attack, and $\texttt{Fail}$ that of a failed attack.}
    \begin{tabular}{|c|c|c|}
        \hline
          & $C_\val$ & $NC_\val$ \\
         \hline
         $\texttt{Succeed}$ & $r+\frac{r}{32}$ & $0$ \\
         \hline
         $\texttt{Fail}$ & $\frac{r}{32}$ & $0$\\
         \hline
    \end{tabular}
    \label{table:pay-off_matrix-strong_simple_game}
\end{table}
%
%
\end{proof}

\section{The Extended Attack in the Presence of Solo Validators}
\label{sec:solo-validators-extended}

We next describe the extended attack,
where the adversary engages with the validators of multiple consecutive slots to reorg multiple blocks.

\subsection{The Extended Game}
\label{sec:repeated-game}

The extended game starts at the beginning of some slot $t+p$ at time $3(t+p)$, where the leader of slot $t+2p+1$ is adversarial.
We normalize $t+p = 0$. 
Let $B_i, i \in \{-p,\ldots, 0\}$, denote the block from slot $i$ (the non-normalized slot $t+i$).
Suppose at the end of slot $0$, the canonical chain ends with the block sequence $B_{-p}, \ldots, B_0$, each block with $W$ votes.
The game has 
the attestors and the leaders of the slots in $[p]$\footnote{We define notation $[p]$, $p \in \mathbb{N}$, as $[p] = \{1,2, \ldots, p\}.$} 
as its players (excluding the adversarial leader of slot $p+1$).
All players are rational solo validators.
The actions available to the players are as in Section~\ref{sec:simple-game}.

The adversary \emph{succeeds} if the canonical chain at the end of slot $p$ is a sequence of empty blocks extending $B_{-p}$, \ie, if the chain consisting of the blocks $B_{-p+1}, \ldots, B_0$ has been reorged by a chain of empty blocks $B^e_{1}, \ldots, B^e_p$ from slots $1, \ldots, p$ (Fig.~\ref{fig:proposal-rule-extended-game}).
It \emph{fails} otherwise.
The game ends, and the payoffs are realized at the end of slot $p$.

\subsection{The Adversarial Strategy}
\label{sec:repeated-game-adversarial-action}

\begin{algorithm}[b]
    \captionsetup{font=small} 
    \caption{The algorithm used by the leader and the committee members of a slot $i \in [p]$ to identify \emph{the compliant tip}. Let $\T$ denote the tree of blocks that includes $B_{-p}$ and its descendants 
    at time $3i$. Let $B.W$ denote the weight of a block $B$ as determined by the votes for $B$ and its descendants. The function $\textsc{LMD-GHOST}$ takes a blocktree as input and returns the canonical chain identified by the LMD-GHOST fork-choice algorithm. The function $\textsc{Index}$ takes a chain of blocks as input and returns the slot number of the non-compliant block with the largest slot within the chain.}
    \label{alg.compliant}
    \begin{algorithmic}[1]\small
    \Function{\textsc{ProposalCompliant}}{$\T$}
        \Let{S}{\{\}}
        \For{$B \in \T$}
        \label{line:iterate}
            \Let{B.W}{B.W + (p-i+1) W + \Wp}
            \label{line:add}
            \Let{\chain_B}{\textsc{LMD-GHOST}(\T)}
            \label{line:lmd-ghost}
            \If{$B \in \chain_B$}
            \label{line:check}
                \Let{S[B]}{\textsc{Index}(\chain_B)}
                \label{line:record}
            \EndIf
            \Let{B.W}{B.W - (p-i+1) W - \Wp}~\Comment{Restore $\T$.}
        \EndFor
        \Let{B^*}{\argmin_{B}(S[B])}
    \EndFunction
    \end{algorithmic}
\end{algorithm}

Before describing the adversary's strategy, we define \emph{compliant} blocks and votes (Fig.~\ref{fig:proposal-rule-extended-game}, Def.~\ref{def:compliant-blocks}). 
To this end, we first define the concept of \emph{compliant} tip at slot~$i$ (\cf~Alg.~\ref{alg.compliant}).
Consider the leader $L_i$ of a slot $i \in [p]$ that wishes to propose a compliant block at time $3i$.
To identify the compliant tip at time $3i$, $L_i$ 
iterates over all blocks $B$ in the block tree of $B_{-p}$ and its descendants at time $3i$ (Alg.~\ref{alg.compliant}, Line~\ref{line:iterate}).
For each block $B$, it envisions the hypothetical canonical chain $\chain_B$ that would be returned by the LMD-GHOST fork-choice rule at the start of slot $p+1$, if \emph{all the attestors of slots $i, \ldots, p$ (in total $(p-i+1) W$) were to vote for $B$'s descendants (but not $B$), and $\mathcal{A}$'s block $B_A$ were to extend the tip of the chain descended from $B$} (Alg.~\ref{alg.compliant}, Lines~\ref{line:add} and~\ref{line:lmd-ghost}).
It subsequently checks if $B$ would be part of this canonical chain $\chain_B$ (Alg.~\ref{alg.compliant}, Line~\ref{line:check}).
If so, $L_i$ records the slot number of the non-compliant block with the largest slot in $B$'s prefix in a dictionary $S$ indexed by blocks (Alg.~\ref{alg.compliant}, Line~\ref{line:record}).
After repeating this procedure for all blocks in the block tree of $B_{-p}$, $L_i$ identifies the block $B^*$ such that $B^*$ is in $\chain_{B^*}$ and the slot number of the last \emph{non-compliant} block in $B^*$'s prefix is smallest among all blocks considered by Alg.~\ref{alg.compliant} (ties are broken deterministically).
We call block $B^*$ \emph{the compliant tip} at slot $i$.
The same steps are used by the slot $i$ attestors to identify the compliant tip.

\begin{definition}[Compliance]
\label{def:compliant-blocks}
\begin{figure}[t!]
    \centering
    \includegraphics[height=1.5in]{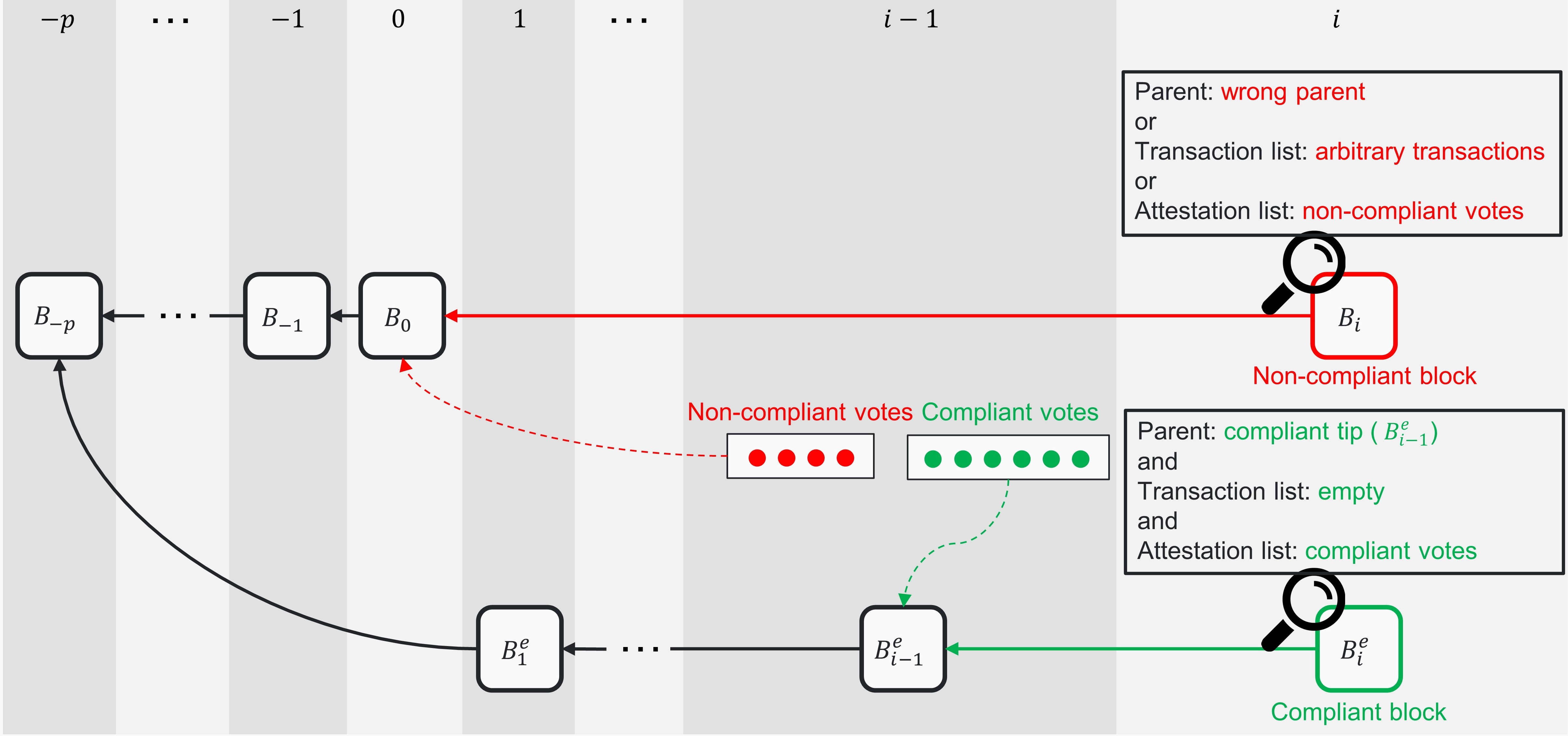}
    \caption{Compliant blocks and votes in the extended game. The extended game is successful if the chain of blocks $B_{-p+1}, \ldots, B_0$ gets reorged by the chain of empty blocks $B^e_{1}, \ldots, B^e_p$.}
    \label{fig:proposal-rule-extended-game}
\end{figure}
A compliant slot $i \in [p]$ block is denoted by $B^e_i$ ($B^e_0 = B_{-p}$).
\begin{itemize} [leftmargin=*]
    \item A compliant slot $1$ block is \emph{empty}, and is proposed on top of block $B_{-p}$.
    \item A slot $i$ vote is compliant if it is for the compliant tip returned by Alg.~\ref{alg.compliant} at time $3i + 1$.
    \item For slots $i = 2,\ldots,p$, a slot $i$ block is compliant if it is \emph{empty}, is proposed on top of the compliant tip returned by Alg.~\ref{alg.compliant} at time $3i$, and contains \emph{only} the compliant slot $i-1$ votes for its parent.
\end{itemize}
\end{definition}
Validators observe the protocol throughout the slots in $[p]$ and iteratively judge whether a block or vote is compliant according to the compliant tip of the slots in $[p]$ as time proceeds.
Thus, they can identify the last non-compliant block at Alg.~\ref{alg.compliant}, Line~\ref{line:record}.

%

\begin{definition}[Adversary's Strategy for the Extended Game]
\label{def:adversary-action-extended}
    Before slot $1$ starts, the adversary $\mathcal{A}$
    \begin{enumerate}[leftmargin=*]
        \item Generates a proof of its leadership in slot $p+1$,
        \item Sets the game rule as follows: For each $i \in [p]$, the slot $i$ leader should propose a compliant block, and the slot $i$ attestors should send compliant votes.
        \item Commits to excluding non-compliant slot $p$ votes from its block $B_A$.
        \item Broadcasts the leadership proof and the game rule to all validators of slots in $[p]$.
    \end{enumerate}
    
    At the start of slot $p+1$, $\mathcal{A}$
    \begin{enumerate}[leftmargin=*]
        \item Proposes its block $B_A$ on top of the compliant tip returned by Alg.~\ref{alg.compliant} at time $3p$.
        \item Includes \emph{only} the compliant slot $p$ votes in $B_A$.
    \end{enumerate}
\end{definition}

\subsection{Analysis}
\label{sec:repeated-game-analysis}
%
As the extended game lasts multiple slots and the future players observe the past actions, we show a subgame perfect Nash equilibrium (SPNE), where solo players follow the game rule.
\begin{theorem}
\label{thm:repeated-game-theorem}
There exists a SPNE of the extended game in the presence of solo validators, where the adversary succeeds 
and violates the security of LMD GHOST.
\end{theorem}
Proof of Theorem~\ref{thm:repeated-game-theorem} is provided in Appendix~\ref{sec:appendix-proof-extended-game}.
It shows that all players are incentivized to follow the adversary's game rule rather than the honest protocol at each subgame corresponding to a slot in $[p]$.
Therefore, they create a chain of empty blocks $B^e_{1}, \ldots, B^e_p$ conflicting with the chain $B_{-p+1}, \ldots, B_0$, and thus violating LMD GHOST's safety and liveness\footnote{The attack considers solo validators, and as such, would work only within an epoch, \ie, for $2p \leq 32$. Even in the absence of a confirmation rule for LMD GHOST, the attack violates the liveness of Casper FFG by causing arbitrarily long reorgs within an epoch.}.

\begin{note}
\label{note:5}
(i) There are many SPNE, where the adversary $\mathcal{A}$ is not successful, 
and (ii) when the number of honest validators per slot is $W_h<(W/2)$ (honest minority), LMD GHOST is still susceptible to the extended attack. 
\end{note}

\section{Attacks in the Presence of Staking Pools}
\label{sec:staking-pools}
In this section, we analyze the attacks in the presence of staking pools.
We model a staking pool as a group of colluding rational validators across different slots.
We call a staking pool that controls $W_s$ attestors at each slot a $W_s$-staking pool.
Staking pools can react differently towards the adversary's committed actions than solo validators.
For instance, 
analyzing the simple game with solo validators, we had to consider only the reward that a solo validator received for its slot $t$ vote.
However, in the presence of staking pools, we should consider the reward in both slots $t-1$ and $t$ as the staking pool controls attestors in both slots.
Indeed, if the adversary reorgs the slot $t$ block, the rewards of the staking pool for its slot $t-1$ votes are affected as well.

\subsection{The Simple Game}
\label{sec:fixed-validator-set-simple-game}
The simple game and the adversarial strategy are the same as in Section~\ref{sec:simple-game} except that a staking pool's payoff depends on its votes in each of the slots $t-1$ and $t$.
The players other than the staking pool are solo attestors.

\begin{theorem}
\label{thm:fixed-validator-set-simple-game-theorem}
Given a $W_s$-staking pool, where $W_s < \Wp$, there exists a Nash equilibrium of the simple game, where the adversary succeeds.
\end{theorem}
Proof of Theorem~\ref{thm:fixed-validator-set-simple-game-theorem} is presented in Appendix~\ref{sec:appendix-proof-simple-game-fixed} and is very similar to the proof for the simple game with solo validators as $W_s < \Wp$. Fig.~\ref{fig:simple_game_staking_pool} illustrates the intuition behind the proof.

Note that if the staking pool controls over $\Wp$ attestors\footnote{As of October 23, 2024, the largest staking pool is LIDO, holding $27.8\%<\Wp = 40\%$ of the total Ethereum stake~\cite{lido}.} per slot, acting non-compliantly 
weakly dominates any other action in the simple game, even though voting 
compliantly still remains weakly dominant.
Therefore, we can still establish the existence of weak Nash equilibrium of the simple game in the presence of a staking pool \emph{with any size}, where the adversary succeeds.
\begin{figure*}[t]
    \centering
    \begin{subfigure}{0.45\textwidth}
        \centering
        \includegraphics[height=1in]{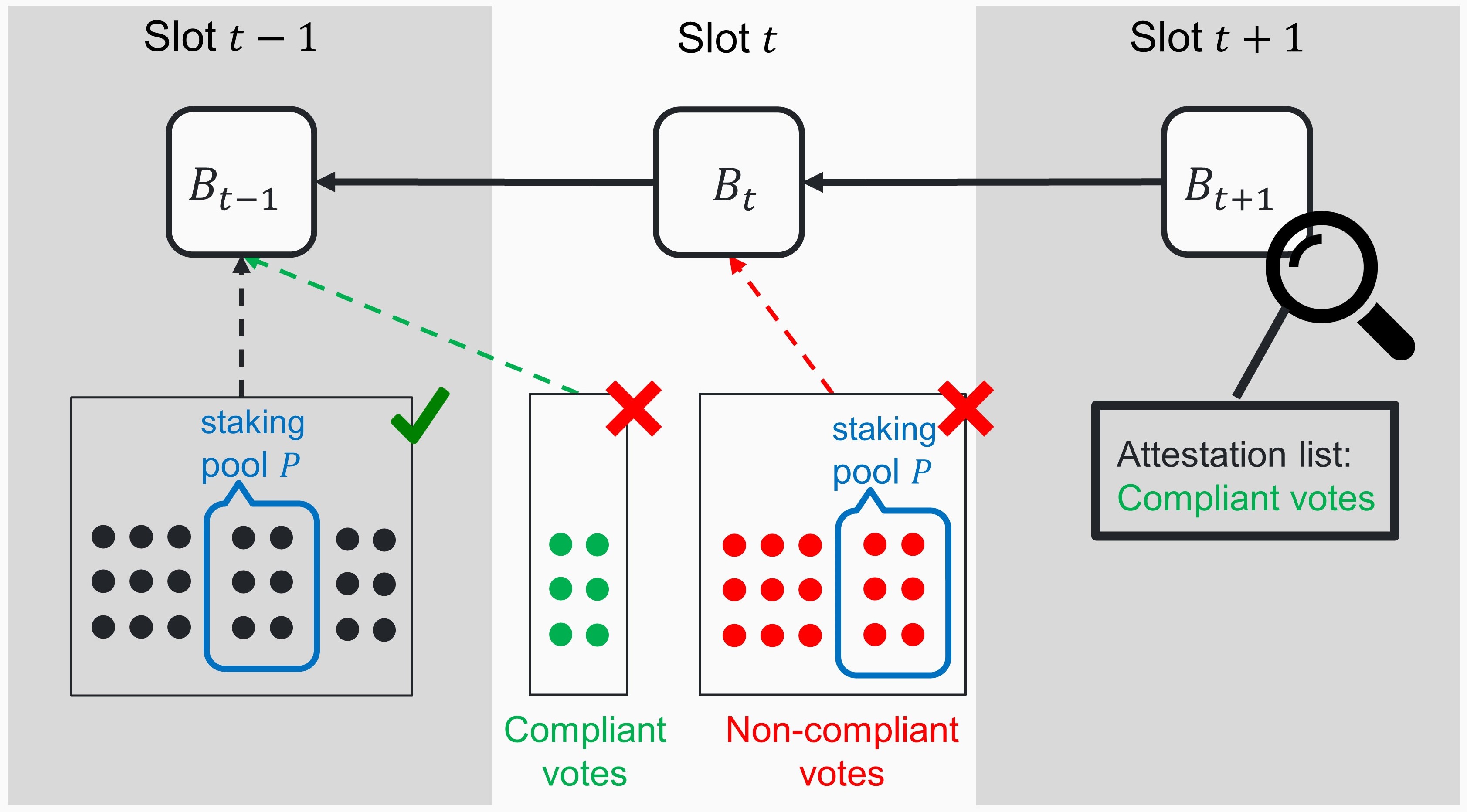}
        \caption{$\texttt{Fail}$, $NC_P$. The payoff of staking pool $P$ is equal to $mr$.}
        \label{fig:simple_game_staking_valool_AV}
    \end{subfigure}%
    ~ 
    \begin{subfigure}{0.45\textwidth}
        \centering
        \includegraphics[height=1in]{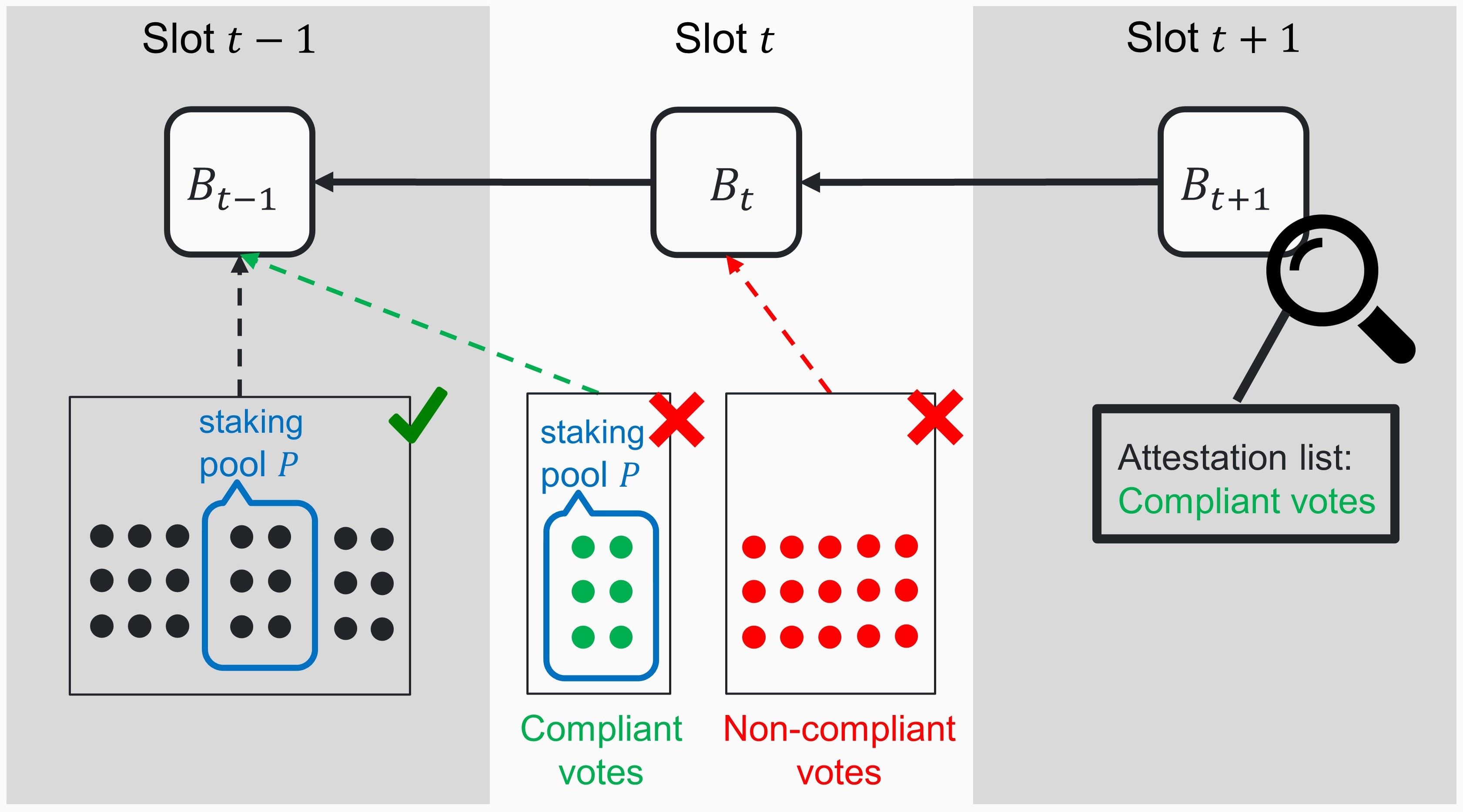}
        \caption{$\texttt{Fail}$, $C_P$. The payoff of staking pool $P$ is equal to $mr$.}
        \label{fig:simple_game__staking_valool_ANV}
    \end{subfigure}
    \vskip\baselineskip
    \begin{subfigure}{0.45\textwidth}
        \centering
        \includegraphics[height=1in]{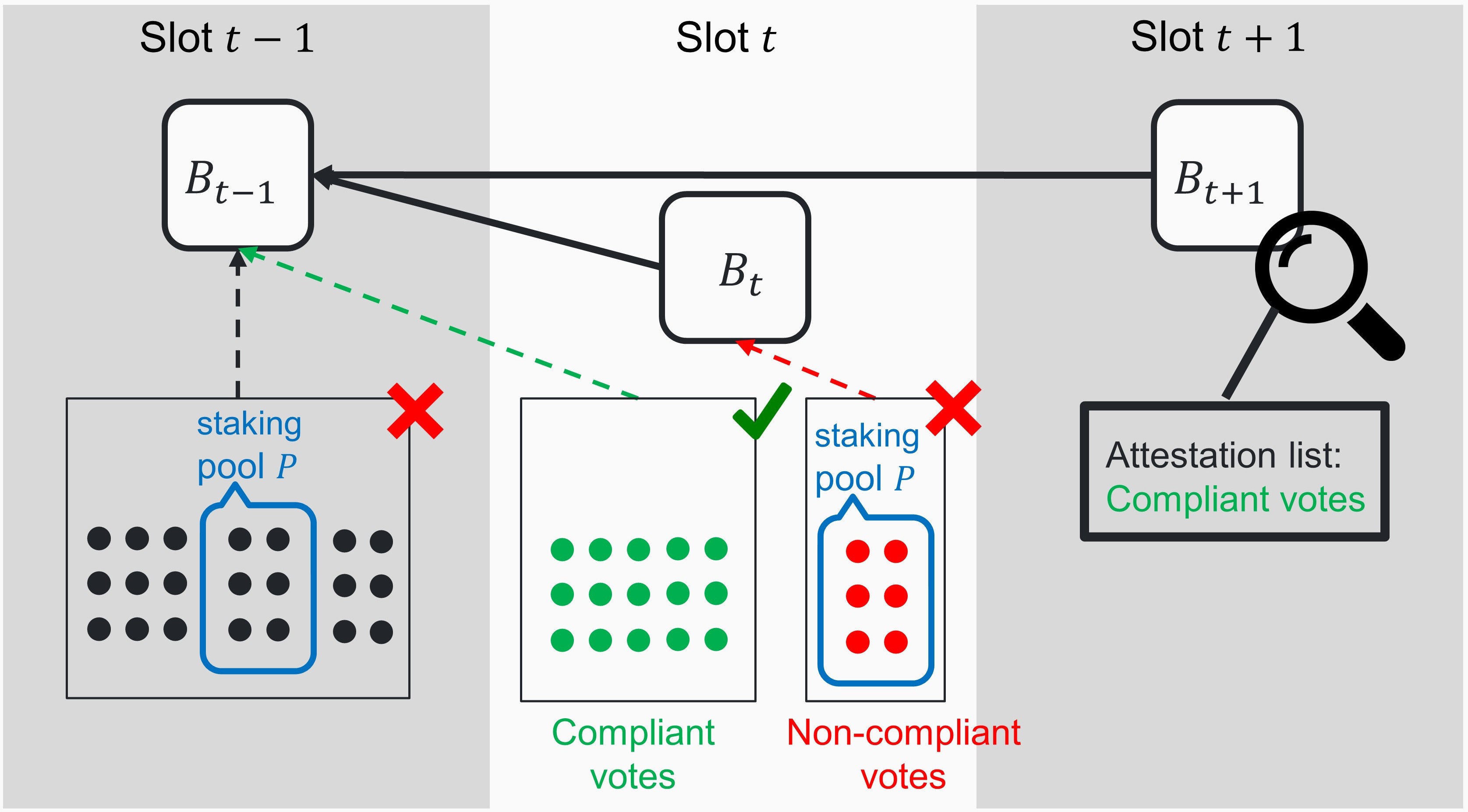}
        \caption{$\texttt{Succeed}$, $NC_P$. The payoff of staking pool $P$ is equal to $0$.}
        \label{fig: simple_game__staking_valool_NAV}
    \end{subfigure}
    ~ 
    \begin{subfigure}{0.45\textwidth}
        \centering
        \includegraphics[height=1in]{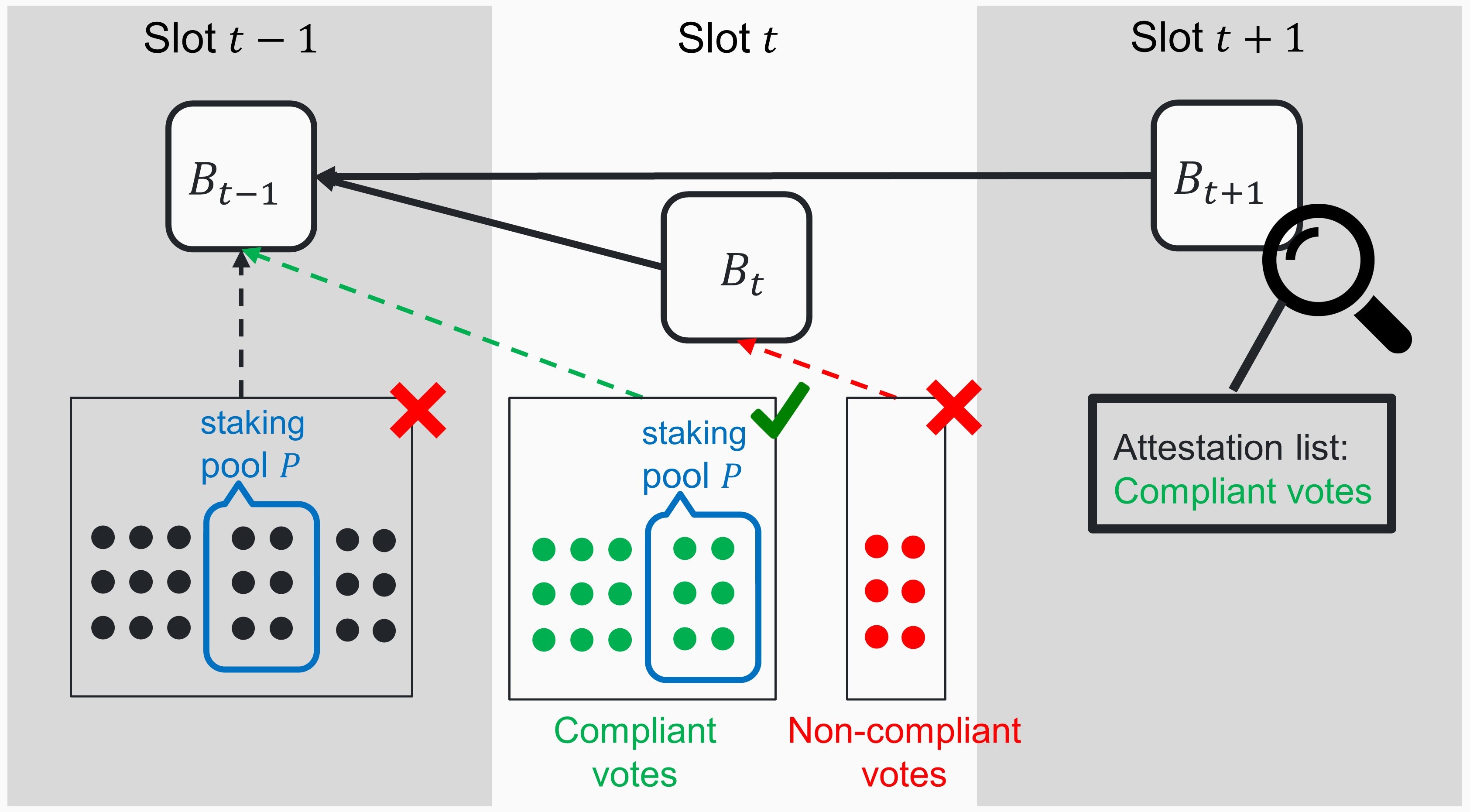}
        \caption{$\texttt{Succeed}$, $C_P$. The payoff of staking pool $P$ is equal to $mr$.}
        \label{fig:simple_game_staking_valool_NANV}
    \end{subfigure}%
    \caption{The payoff of a staking pool $P$ with a total stake share less than $\Wp$ in the simple game. $P$ controls $m$ votes in each slot, where $m \geq 1$.
    $r$ represents the reward for a single correct and timely vote. 
    Let $C_P$ and $NC_P$ denote the events that the staking pool $P$ votes and does not vote for $B_{t-1}$ (acts compliantly and non-compliantly), respectively. Let $\texttt{Succeed}$ denote the event that less than $\Wp$ slot $t$ attestors vote for $B_t$ (the attack succeeds), and $\texttt{Fail}$ the event that $\Wp$ or more slot $t$ attestors vote for $B_t$ (the attack fails). The green (red) votes are compliant (non-compliant). The green tick and the red cross represent whether a vote receives or does not receive a reward, respectively. As can be seen, $C_P$ (compliance) weakly dominates $NC_P$ (non-compliance).}
    \label{fig:simple_game_staking_pool}
\end{figure*}

\subsection{The Extended Game}
\label{sec:fixed-validator-set-extended-game}
The extended game and the adversarial strategy are the same as in Section~\ref{sec:solo-validators-extended} except that a staking pool's payoff depends on its votes in each of the slots.
The players other than the staking pool are solo attestors and leaders.
%
\begin{theorem}
\label{thm:fixed-validator-set-extended-game-theorem}
Given a $W_s$-staking pool, where $W_s < W_p$, there exists a SPNE of the extended game, where the adversary succeeds and violates the security of LMD GHOST.
\end{theorem}
Proof of Theorem~\ref{thm:fixed-validator-set-extended-game-theorem} is in Appendix~\ref{sec:appendix-proof-extended-game-fixed} and has a similar outline as the proof of Theorem~\ref{thm:repeated-game-theorem}.

If the staking pool controls over $\Wp$ attestors, compliance with the adversary's game rule ceases to be a dominant strategy, leading to the mitigation of the attack. 
In the next section, we introduce the \emph{selfish mining-inspired attack}, where even colluding validators are incentivized to comply with the adversary.
In the presence of staking pools of any size, it threatens the protocol's liveness regardless of the size of the staking pool, albeit requiring a larger share of the stake.

\subsection{Selfish Mining-Inspired Attack} \label{sec:selfish_mining_game}

In this section, we present a new, selfish mining-inspired attack that enables the adversary to do a long-range reorg attack even in the presence of large staking pools.
Although the attack cannot violate safety or liveness permanently when the adversary's stake is $<1/3$, over short periods where adversarial slots outnumber honest slots, it enables the reorg of non-adversarial blocks as in the original selfish-mining attack.
Thus, it is an attack on reorg-resilience.

\begin{figure}[b]
    \centering
    \includegraphics[height=1in]{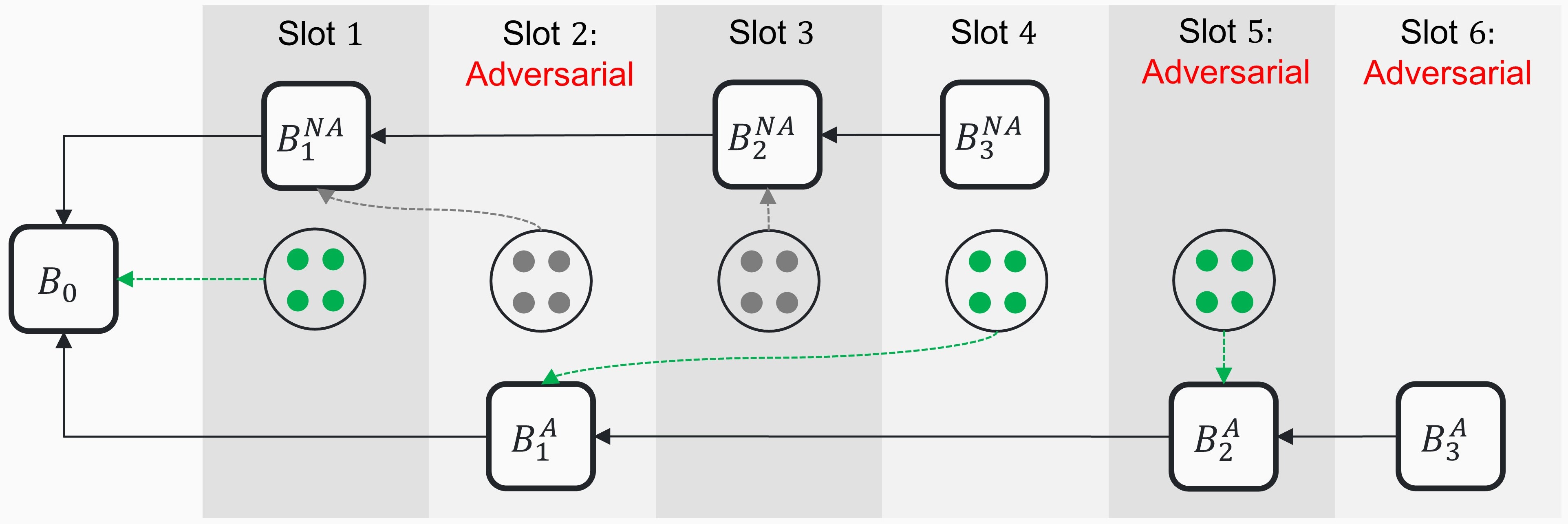}
    \caption{The selfish mining-inspired attack, where $p = 5$ and $N_\mathcal{A}=N_\mathcal{NA}=3$. The upper and the lower forks are non-adversarial and adversarial, respectively. The attestors in slots 1, 4, and 5 are incentivized to vote for the adversarial fork as their subsequent slots are adversarial. At the start of slot $6$, the weights of the non-adversarial and adversarial forks are $2W$ and $2W+\Wp$, respectively, resulting in the adversarial fork winning.}
    \label{fig:selfish-mining}
\end{figure}

%
The \emph{selfish mining-inspired attack} starts at the beginning of some slot $t$, where the number of adversarial slots in the interval $t, \ldots, t+p$ is greater than or equal to the number of its non-adversarial slots, and slot $t+p$ is adversarial.
We normalize slot $t-1 = 0$, implying that the attack occurs over the slots $[p+1]$.
The game has the attestors of those slots preceding the adversarial slots in $[p+1]$ as its players. 
Let $N_\mathcal{A}$ and $N_\mathcal{NA}$ denote the number of adversarial and non-adversarial slots within the slot range $[p+1]$.
Let $B_0$ denote the block at the tip of the canonical LMD GHOST chain at the end of slot $0$.
Let $s_i$ and $B^\mathcal{A}_i$, $i \in \{1, \ldots, N_\mathcal{A}\}$, denote the $i^\text{th}$ adversarial slot within the range $[p+1]$ and its corresponding block, respectively.
We have $N_\mathcal{A} \ge N_\mathcal{NA}$. In this attack, the adversary commits to the following strategy:
\begin{definition}[Adversary's Strategy for the Selfish Mining-Inspired Game]
\label{def:selfish_mining_game}
    Before slot $1$ starts, the adversary $\mathcal{A}$ generates a proof asserting its leadership in the adversarial slots within the range $[p+1]$.
    Then, for the slots $s_i-1$, $i \in \{1, \ldots, N_\mathcal{A}\}$, it sets the suggested game rule $\texttt{Game}_i$ as follows: 
    The attestors of slot $s_i-1$ should refrain from publishing their votes at time $3(s_i-1)+1$ (\ie, should withhold their votes).
    They should later vote for block $B^\mathcal{A}_{i-1}$ that will be published by the adversary at a point in the time range $[3p, 3(p+1)+1]$\footnote{The time range is after all the non-adversarial blocks within $[p+1]$ have been published and before the voting time of slot $p+1$ starts.}.
    The adversary does not specify any rules for the attestors of the remaining slots.

    At the start of slot $p$, when all the non-adversarial blocks within $[p+1]$ have been published, the adversary $\mathcal{A}$ takes the following action for $i \in \{1, \ldots, N_\mathcal{A}\}$:
    \begin{itemize}[leftmargin=*]
        \item It publishes block $B^\mathcal{A}_{i-1}$ for the attestors of slot $s_i-1$ ($B^\mathcal{A}_0$ is defined to be $B_0$).
        \item It collects the votes submitted by the attestors of slot $s_i-1$ for block $B^\mathcal{A}_{i-1}$ and inserts them into the block $B^\mathcal{A}_{i}$ as compliant votes.
    \end{itemize}
    Before time $3(p+1)+1$, the adversary publishes the sequence of adversarial blocks $B^\mathcal{A}_{i}$, $i \in [N_\mathcal{A}]$ proposed on top of block $B_0$ for all the validators.
    
\end{definition}

The attack \emph{succeeds} if the canonical chain at the end of slot $p+1$ is the sequence of adversarial blocks $(B^\mathcal{A}_{1}, \ldots, B^\mathcal{A}_{N_{\mathcal{A}}})$ proposed on top of block $B_0$.
The attack \emph{fails} otherwise.
The game ends, and the payoffs are realized before the voting time of slot $p+1$ starts at time $3(p+1)+1$.
Similar to the selfish mining attack in longest-chain blockchains, the adversary in this attack hides the adversarial fork until the end of the attack interval. The adversarial fork publication is delayed to prevent non-adversarial slot leaders from proposing their blocks on top of it, which would allow them to avoid being reorganized. Fig.~\ref{fig:selfish-mining} illustrates the selfish mining-inspired attack conducted over six consecutive slots.

\subsection{Analysis of Selfish Mining-Inspired Attack}
\label{sec:selfish_mining_game-analysis}
Let the non-adversarial fork be the fork containing blocks published by the non-adversarial validators by the end of slot $p$, and the adversarial fork be the sequence $(B^\mathcal{A}_{1}, \ldots, B^\mathcal{A}_{N_{\mathcal{A}}})$. 
If the attestors of all slots $s_i-1$, $i \in [N_\mathcal{A}]$ are incentivized to follow the adversary's game rule rather than the honest protocol, the weights of the non-adversarial fork and the adversarial fork at time $3(p+1)+1$ will be equal to $N_\mathcal{NA}W$ and $N_\mathcal{A}W+\Wp$, respectively.
Then, if $N_\mathcal{A}\ge N_\mathcal{NA}$, the weight of the adversarial fork will be higher, and the attack will succeed. 
\begin{theorem}
\label{thm::selfish_mining_game}
There exists a Nash equilibrium of the selfish mining-inspired game in the presence of staking pools, where the adversary succeeds.
\end{theorem}

\begin{proof}
    Let $S_\mathcal{A}$ denote the set of slots whose next slot is adversarial, and let $S_\mathcal{NA}$ denote the set of slots whose next slot is non-adversarial. Let the staking pool $P$ control $m_s$ attestors in slot $s$. Note that the adversarial game rule implies that the staking pool should vote on the adversarial fork only in slots preceding adversarial slots, meaning that in other slots, the staking pool is free to vote on non-adversarial forks. The payoff matrix of a staking pool $P$ in the selfish mining-inspired game is presented in Table~\ref{table:pay-off-matrix_selfish_mining_game}.
    \begin{table}[b]
    \centering
    \caption{Pay-off matrix of a staking pool $P$ in the selfish mining-inspired game. Let $C_P$ and $NC_P$ denote the events that the staking pool $P$ follows the adversarial game rule (being compliant) and the honest protocol (being non-compliant), respectively. Let $\texttt{Succeed}$ denote that the adversarial fork at time $3(p+1)+1$ outweighs the non-adversarial fork. Conversely, $\texttt{Fail}$ denotes that the non-adversarial fork at time $3(p+1)+1$ outweighs the adversarial fork.}
    \begin{tabular}{|c|c|c|}
         \hline
         & $C_P$ & $NC_P$  \\
         \hline
         $\texttt{Succeed}$ & $\sum_{s \in S_\mathcal{A}}{m_sr}$ & $0$ \\
         \hline
         $\texttt{Fail}$ & $\sum_{s \in S_\mathcal{NA}}{m_sr}$ & $\sum_{s \in S_\mathcal{NA}}{m_sr}$ \\
         \hline
    \end{tabular}
    \label{table:pay-off-matrix_selfish_mining_game}
\end{table}
Following the game rule weakly dominates following the honest protocol. Therefore, a rational staking pool follows the game rule, resulting in the adversary's success.
\end{proof}

\begin{note}
There are Nash equilibria, where the adversary $\mathcal{A}$ is not successful.
However, if in a slot interval, the number of adversarial slots is strictly greater than the number of non-adversarial ones, the selfish mining-inspired attack gives a higher payoff to the attestors in the equilibria it succeeds, as with high probability we have: $\sum_{s \in S_\mathcal{A}}{m_sr}> \sum_{s \in S_\mathcal{NA}}{m_sr}$.
This implies that unlike in the extended attack, even coordination of strategies and forming stake pools might not be sufficient to disincentivize validators from following the game rules.
\end{note}


\section{\NewMechanism Mechanism}
\label{sec:solution}


The \newmechanism mechanism (Fig.~\ref{fig:new-mechanism}) requires each slot $t+1$ attestor to sign all slot $t$ head votes in its view at the beginning of slot $t+1$.
A signature of a slot $t+1$ attestor on a slot $t$ vote is called a \emph{slot $t$ evidence}.
A slot $t$ attestor gains an attestation reward for its vote if the vote is included in \emph{any} slot $t'>t$ block $B$ within the canonical chain, and the vote is \emph{correct} and \emph{timely} with respect to $B$.
Correctness and timeliness are re-defined as follows:
%

\smallskip
\noindent
\textbf{Correctness:}
The slot $t$ vote is correct if it is for the highest block from slots $\leq t$ in the prefix of $B$ (let $B_\ell$ denote this block). 

\smallskip
\noindent
\textbf{Timeliness:}
The slot $t$ vote is timely if either $B'$ is of slot $t+1$, or there are over $W/2$ unique slot $t$ evidences for the slot $t$ vote in the prefix of $B'$ after block $B_\ell$.

As the mechanism relies on slot $t+1$'s committee to infer whether a slot $t$ vote is timely, it reduces the adversary's influence over the attestation rewards, unless it controls half of the committee.
An adversary with bounded stake cannot delay the inclusion of evidences in the chain indefinitely, implying that its previous attacks fail (\cf Sections~\ref{sec:solo-validators-simple},~\ref{sec:solo-validators-extended} and~\ref{sec:selfish_mining_game}).
The \newmechanism mechanism also allows rewarding the slot $t+1$ committee for generating slot $t$ evidences and the leaders for including the past votes and evidences.

%
%

\subsection{Analysis of the \NewMechanism Mechanism}
\label{sec:dag-votes-analysis}
To show the effectiveness of \newmechanism, we prove the existence of a subgame perfect Nash equilibrium in an idealized model, in which LMD GHOST equipped with \newmechanism remains secure under the following assumptions:
\begin{enumerate}
    \item A slot $t$ vote is considered `timely', only if there are over $W/2$ unique slot $t$ evidences for the vote, where $W$ is the number of attestors per slot.
    This assumption prevents a slot $t+1$ block $B$ that was not proposed at the tip of the canonical LMD GHOST chain (called `canonical chain' from now on for brevity) from rewarding `late-released' slot $t$ votes, when $B$ later enters the canonical chain. 
    \item All rational, non-colluding (\ie, solo) slot $t+1$ attestors create evidences for the slot $t$ votes they observe by $2\Delta$ time into slot $t$ (\ie, $\Delta$ time before the start of slot $t+1$) and broadcast these votes to all other attestors (they do not create evidences for the votes observed later). This is necessary to avoid penalizing the votes broadcast in a timely manner.
    \item If a vote by an attestors is not detected as a `timely' vote, that attestor is eventually penalized. This captures the inactivity leak functionality~\cite{inactivity-leak} ingrained in Ethereum consensus and is a necessary assumption for proving Theorem~\ref{thm:dag-votes-secure}.
    \item Rational leaders broadcast their blocks at the beginning of their respective slots\footnote{Timing games, where leaders delay their blocks for financial gain, are beyond the scope of this work. For more discussion, \cf~Section~\ref{sec:related-work-timing-games}}.
    \item There is no proposer boost\footnote{Although LMD GHOST is not secure against Byzantine attestors in the absence of proposer boost, it is not necessary for Theorem~\ref{thm:dag-votes-secure} that considers rational attestors.}. If instead the proposer boost is set to $W_p>0$, Theorem~\ref{thm:dag-votes-secure} would require the number of solo attestors to be at least $1+((W+W_p)/2)$.
\end{enumerate}
We note that the equilibrium strategies described below constitute a strong Nash equilibrium within each subgame of the extended game, up to collusions of size $(W/2)-1$.
\begin{theorem}
\label{thm:dag-votes-secure}
Consider executions of LMD GHOST with \newmechanism, where all attestors are rational, and at any given slot, over $1+(W/2)$ attestors are solo attestors.
Suppose that the slot leaders are rational, except the leader of some slot $t+1$, which is adversarial.
Then, there exists a subgame perfect Nash equilibrium, where the protocol satisfies security (\ie, safety and liveness).
%
\end{theorem}
\begin{proof}
We will show the existence of a subgame perfect Nash equilibrium (SPNE), where the solo attestors vote for the block at the tip of the canonical chain at the voting time of their respective slots, and the rational slot leaders create blocks at the tip of the canonical chain at all slots.
In this setting, the adversarial leader cannot cause a reorg of the canonical chain, as its block will fail to gather any votes, if it is not placed at the tip of the canonical chain at slot $t+1$.

For backward induction, suppose that rational leaders and attestors after some slot $s$ follow the prescribed strategy above.
%
\begin{lemma}
\label{lem:all-observe}
For any slot $s$, votes of all attestors of the slots less than $s$ are observed by the slot $s$ attestors by the start of slot $s$.
Moreover, all attestors send timely votes.
\end{lemma}
\begin{proof}
By assumption (1), a slot $s'<s$ vote is considered timely only if there are over $W/2$ unique slot $s'$ evidences by the slot $s'+1 \leq s$ attestors.
Now, by assumption (3), a rational attestor incurs a loss if its vote is not considered timely.
In contrast, its payoff for a timely vote is zero in the worst case.
More specifically, a rational attestor gains a positive reward if its vote is both correct and timely; if the vote is timely but not correct, its payoff becomes zero.
Therefore, in all slots $s'<s$ (as well as later slots), rational attestors ensure that their vote will be considered timely, \ie, there will be more than $W/2$ unique slot $s'$ evidences from the attestors of slot $s'+1$.
This implies that a solo slot $s'+1$ attestor must have observed the slot $s'$ vote $\Delta$ time before the start of slot $s'+1$, which by assumption (2), is broadcast to and heard by all other attestors by the start of slot $s'+1 \leq s$.
Then, as all attestors are rational, votes of all attestors of the slots less than $s$ are observed by the slot $s$ attestors by the start of slot $s$.
\end{proof}
\noindent
\textbf{Subgame for slot $s$ attestors:} By Lemma~\ref{lem:all-observe}, all solo slot $s$ attestors observe all votes sent for slots $<s$ by the start of slot $s$ (\ie, all past votes are public).
Therefore, when solo slot $s$ attestors vote for the tip of the canonical chain at the voting time of slot $s$, they end up voting for the same block $B$, which is at the tip of the canonical chain according to the votes from slots $\leq s$.
Now, by the assumption that later slot leaders and attestors follow the prescribed strategy, block $B$ and its descendants gain over $1+(W/2)$ votes at each slot $s' \in [s,t]$, thus remain in the canonical chain.
We note that even if $s=t$, the adversary cannot reorg $B$; since (i) by the start of slot $t+1$, all votes of the previous slots will be public by Lemma~\ref{lem:all-observe}, and (ii) $B$ will be at the tip of the canonical chain with an advantage of at least one vote at that time, implying that the solo attestors of slot $t+1$ will not vote for the adversary's block unless it extends $B$.

Now, by assumption (4), block $B$ is the latest block on the canonical chain from slots $\leq s$.
Thus, the votes of the solo slot $s$ attestors are correct (and timely by Lemma~\ref{lem:all-observe}), implying that they receive a reward of $r>0$ for their votes.
In contrast, if a solo slot $s$ attestor does not vote for $B$, its vote would not be correct, implying a reward of $\leq 0$.
Hence, there exists an equilibrium of the subgame of the slot $s$ attestors, where all solo slot $s$ attestors send timely votes for the canonical chain tip in their views.

\smallskip
\noindent
\textbf{Subgame for slot $s$ leader:} Suppose all attestors of the slots $s' \geq s$ and the leaders of the slots $s' > s$ follow the prescribed strategy above, and consider the payoff of the slot $s$ leader.
By Lemma~\ref{lem:all-observe}, the slot $s$ leader observes all votes sent for slots $<s$ by the start of slot $s$.
Then, if the leader places its block $B$ at the tip of the canonical chain, by the assumption on the slot $s' \geq s$ attestors and the slot $s'>s$ leaders, $B$ and its descendants remain part of the canonical chain, implying a reward of $R>0$ for the leader (see the reasoning in the paragraph above).
Otherwise, the leader's reward would be $0$.
Therefore, there exists an equilibrium of the subgame of the slot $s$ leader, where it proposes the slot $s$ block at the tip of the canonical chain in its view.

Finally, by backward induction, all solo attestors vote for the block at the tip of the canonical chain, and rational slot leaders create blocks at the tip of the chain in all slots up to the slot $t+1$.
Moreover, all of these blocks remain as part of the canonical chain.
This implies that the prescribed strategy is an SPNE, where the protocol satisfies safety and liveness.
\end{proof}

\subsection{\NewMechanism Made Practical}
\label{sec:practical}
We now present a practical version of \newmechanism and compare its performance with Ethereum's reward mechanism.
Recall that an attestation includes a signature on the head, source checkpoint, and target checkpoint votes.

\subsubsection{ Current version of Ethereum} 
Let $N$ denote the total number of validators\footnote{As of October 20, 2024, there are 1,073,375 active Ethereum validators~\cite{ethereum_info}.}. Each slot has $N/32$ attestors, 
which are divided into $64$ sub-committees, each with $\frac{N}{32*64}$ attestors.
A subset of the sub-committee is selected to be the aggregators of the sub-committee\footnote{Average number of aggregators in each sub-committee is 16.}. 
These aggregators are responsible for aggregating the attestations (signatures) of their sub-committee. 
Ethereum uses the BLS signature for attestations due to its key-homomorphism~\cite{bls}.
In the ideal scenario, where all attestations of the sub-committee members are the same, these attestations can be aggregated into a single aggregated attestation (aggregate for short). 
Each aggregate contains an aggregated signature and an aggregation list. 
The $i^\text{th}$ element in the aggregation list is $1$ if the signature of the $i^\text{th}$ attestor in the corresponding sub-committee is included in the aggregated signature, and $0$ otherwise. 
In the ideal scenario, there exist $64$ aggregates per slot in total. Each block can include up to $128$ aggregates, where
the slack between $64$ and $128$ can cover attestations for different blocks or prior slots\footnote{These attestations may be eligible to get the source/target rewards.}.

\subsubsection{ Practical version of \newmechanism} \label{sec:practical_DAG_votes}
In our solution introduced in Section~\ref{sec:practical_DAG_votes}, all the attestors in each slot are responsible for generating evidences. 
This can increase the communication and storage overheads drastically. 
To address this issue, we propose a practical version of \newmechanism. 
In this version, a subset of attestors is sampled in each slot to serve as the \emph{aggregators} responsible for generating the evidence. In contrast to the main version, only the aggregators produce the evidences in the practical version.

Let sub-committee $(t,i)$ represent the $i^\text{th}$ sub-committee of slot $t$, $N^\text{att}$ denote the number of attestors per sub-committee, $N^\text{agg}$ denote the number of aggregators per sub-committee. An attestation must gather over $N^\text{limit}$ evidences to be considered timely.
Each attestor in sub-committee $(t+1,i)$ sends an attestation at time $3(t+1)+1$. 
Then, the aggregator $j \in [N^\text{agg}]$ in sub-committee $(t+1,i)$ takes the following actions at time $3(t+1)+2$:
\vspace{-2 pt}
\begin{itemize}[leftmargin=*]
    \item Collect the individual attestations of sub-committee $(t+1,i)$ members and aggregate them into a single aggregated attestation (aggregate) denoted by $\mathsf{Ag}^j_{t+1,i}$.
    \item Generate an evidence for slot $t$ attestations if the block of slot $t+1$ is not published or excludes some of the slot $t$ attestations.
    Collect the 
    aggregates by the aggregators of sub-committee $(t,i)$, \ie, $\mathsf{Ag}^k_{t,i}$ for $k \in [N^\text{agg}]$, select and sign the one with the greatest number of signatures 
    to generate the $j^\text{th}$ sub-committee $(t,i)$ evidence denoted by $E^j_{t,i}$.
    \item Broadcast slot $t+1$ aggregate $\mathsf{Ag}^j_{t+1,i}$ and slot $t$ evidence $E^j_{t,i}$.
\end{itemize}
If slot $t$ attestations are not included in the slot $t+1$ block, any block proposed at slots $t'>t+1$ can include the slot $t$ evidences. 
Each evidence is composed of 3 parts: the signature of the evidence generator, the aggregated attestation signature, and an aggregation list, 
specifying which validators' attestations are included in the aggregated signature.

The parameters $N^\text{agg}$ and $N^\text{limit}$ 
should be configured to ensure that each sub-committee has more than $N^\text{limit}$ non-adversarial available aggregators with high probability. 
This ensures the adversary cannot prevent the timeliness verification of attestations and cannot prove delayed votes as timely by publishing their evidences late. 

%
\subsection{Implementation of the DAG Votes} \label{sec:implementation_DAG}
Due to the difficulty of modifying existing testnets for LMD GHOST, we decided to estimate the overhead ratio of evidence creation from limited experiments.
In this context, we have implemented a simplified version of our DAG votes mechanism in Python~\cite{DAG_Vote_Implementation}, utilizing BLS signatures~\cite{blspy} over the BLS12-381 curve. This implementation features a peer-to-peer local blockchain running across multiple Amazon Web Services (AWS) EC2 instances. One instance serves as a Bootstrapping node to initialize the blockchain and assist other nodes in discovering peers.

Our implementation includes two configurations: the Ethereum basic implementation and the DAG votes mechanism. The Ethereum basic implementation replicates the behavior of Ethereum's LMD GHOST protocol, incorporating functionalities such as block proposal, vote proposal, and aggregation. For any arbitrary slot duration of \(3\Delta\) seconds, the slot \(t\) block is proposed at time \(3t\Delta\), the slot \(t\) votes are broadcast at time \((3t+1)\Delta\), and the slot \(t\) aggregations are broadcast at time \((3t+2)\Delta\).
\begin{table*}[th!]
    \centering
    \caption{Comparison of Basic Ethereum Implementation and DAG Votes Mechanism}
    \label{tab:comparison}
    \begin{tabular}{|l|c|c|}
    \hline
    \textbf{Metric}                            & \textbf{Basic Ethereum} & \textbf{DAG Votes Mechanism} \\ \hline
    Block size (excluding transaction payload) & 582 bytes               & 745 bytes                    \\ \hline
    Aggregation size                           & 400 bytes               & 804 bytes                    \\ \hline
    Average aggregation creation time          & 0.0138 s                & 0.0142 s                     \\ \hline
    Average aggregation verification time      & 0.0076 s                & 0.0117 s                     \\ \hline
    \end{tabular}
\end{table*}

The DAG votes mechanism extends this implementation by introducing an evidence generation and verification process to support the DAG votes protocol. In this setup, aggregators not only aggregate votes from the same slot but also create evidence for votes from the previous slot. Specifically, the slot \(t\) aggregation includes the aggregated signature of slot \(t\) votes and a piece of evidence for slot \(t-1\) votes. In this setup, block proposers are tasked with including the aggregated evidence in the blocks they produce.

\noindent \textbf{Comparison of our Ethereum basic implementation and DAG votes mechanism}:  
We compare the performance of the basic Ethereum implementation with the implementation extended by our DAG votes mechanism in Table~\ref{tab:comparison} to analyze the overhead of our solution. 
This comparison includes four metrics: block size, aggregation size, aggregation creation time, and aggregation verification time. The comparison results are obtained by running our implementation on 5 AWS EC2 instances. Each instance is of the \texttt{t3.micro} type, featuring 2 vCPUs and 1.0 GiB of memory. 
At each slot, one validator is randomly selected to propose a block, while all validators submit votes and aggregations. The number of aggregators per slot can be set in our implementation; however, we set it to be equal to the number of validators to ensure all validators create evidence for the previous slot's votes.
A slot $t$ aggregation in our DAG votes mechanism, in addition to including the aggregated votes of the slot $t$ attestors, includes a piece of evidence, which is the aggregator's signature on the aggregated votes of slot $t-1$. The inclusion of this extra piece of evidence compared to basic Ethereum increases the size of the aggregation by 404 bytes. A slot $t+1$ block proposer in our DAG votes mechanism, instead of including the aggregated slot $t$ votes as in basic Ethereum, includes the aggregated pieces of evidence generated in slot $t$. As the aggregated evidence has an extra signature compared to the aggregated votes in basic Ethereum, the size of the block increases by 163 bytes in our solution. 
The average aggregation creation and verification times are calculated based on the time taken by each validator in each slot throughout the local blockchain runtime.


\noindent \textbf{Performance analysis of practical DAG votes deployed in the Ethereum mainnet}:
We provide estimates for the storage, computation, and communication overheads of the practical DAG votes mechanism (introduced in Section~\ref{sec:practical_DAG_votes}) over the current Ethereum protocol if it were to be deployed in the current Ethereum mainnet in two scenarios: (i) the optimistic scenario, where the majority of slot $t+1$ aggregators sign the same slot $t$ aggregate, and ii) the worst-case scenario, where the slot $t+1$ aggregators sign different slot $t$ aggregates, in Appendices~\ref{sec:appendix-optimistic} and~\ref{sec:appendix-worst-case} respectively.

\section{Conclusion}
\label{sec:appendix_stackelberg}

As our attacks show, the ability for adversaries to commit to conditional courses of action presents challenges for the design of protocol mechanisms. 
This holds more so given the availability of commitment devices within the protocol's enclosure, the devices possibly executing the protocol itself. Recent literature has uncovered such attacks on fee market mechanisms, a core function of blockchains~\cite{landis2023stackelberg}, naming the deployment of a smart contract committing the author to certain actions of a ``Stackelberg attack'', after the classic model of leader-followers games~\cite{von1952theory}.
In our attacks, we assumed that the adversary is the first to commit to any strategy and communicate a credible threat to the other parties.
Indeed, the order of commitments is known to be important: players will often gain 
advantage from the ability to commit \textit{first}~\cite{schelling1960strategy, bono2014game}.

The attacks presented in this paper offer a \textit{better} response compared to the honest protocol specifications for some protocol participants, in the presence of commitment devices. We leave for future work deciding whether these attacks constitute \textit{best} responses for rational attackers;  which equilibria exist if any, when all participants benefit from the commitment ability (\eg, proposers attempting to reorg one another in the spirit of undercutting attacks~\cite{carlsten2016instability}); and the importance of the sequence of play to the determination of outcomes.

As Ethereum is an ever-evolving protocol, in Appendix~\ref{sec:appendix_discussion}, we analyze the impact of potential future changes in the protocol (\eg, block-slot voting~\cite{block_slot_voting}, secret leader election~\cite{secretLeaderElection}, proposer-builder separation~\cite{PBS} and single-slot finality~\cite{singleSlotFinality}) on our attacks, and observe that none of these modifications can mitigate them.
Similarly, in Appendix~\ref{sec:appendix-simple-attack-wo-boost}, we present a version of the simple attack without proposer boost, lest the future changes above obviate the need for proposer boost.

\section{Responsible Disclosure}
Since our work exposes vulnerabilities in the existing system, we collaborated closely with researchers at the Ethereum Foundation, a non-profit organization supporting the Ethereum ecosystem. Upon identifying these attack vectors, we promptly contacted the Ethereum Foundation, whose researchers supported our work by providing access to datasets and other resources.

\section{Acknowledgment}
We would like to thank Toni Wahrstätter for his assistance in gathering data used to quantify the MEV reward in Ethereum.

\bibliographystyle{splncs04}
\bibliography{references}

\appendices
\section{Extended Preliminaries}

\subsection{The Ethereum Reward Mechanism}
\label{sec:reward_mechanism_background_appendix}

Validator rewards in Ethereum originate from two sources: the execution layer reward and the consensus layer reward.
The execution layer reward refers to the total reward of a slot leader for including transactions in his block. 
This can be either the sum of all priority fees included in the block or the amount that a block builder pays for buying the leader's block space in a proposer-builder-separation auction~\cite{heimbach2023ethereum}. 
The consensus layer reward includes the incentivization reward 
paid out to the validators for their participation in the consensus protocol execution.
In this paper, 
we focus on the impact of consensus layer reward on the validators' behavior. 
Therefore, we 
describe the consensus layer reward mechanism below.

According to the Altair upgrade~\cite{altair} of the beacon chain, a validator's consensus layer reward is made up of three sub-rewards~\cite{rewards}:

\noindent
    \textbf{Attestation reward:} A validator receives this for attesting to its view of the chain. 
    Let validator $\val$ be a member of the committee assigned to slot $s$.
    It sends an attestation consisting of 3 different votes:
    (i) the vote for a source checkpoint specified by Casper FFG,
    (ii) the vote for a target checkpoint specified by Casper FFG, and
    (iii) the \emph{head vote} for the block specified by LMD GHOST (\cf~Section~\ref{sec:lmd-ghost-background}).
    
    \noindent
    \textbf{Proposing reward:} A slot leader receives this for proposing a block if the proposed block gets included in the beacon chain. 
    
    \noindent
    \textbf{Sync committee reward:} Every 256 epochs, a group of 512 validators is randomly selected as the members of the sync committee. 
    They sign every beacon chain block header published during the assigned 256 epochs to help light clients track the beacon chain. 
    Sync committee members receive a reward for every slot in which they sign a block.


The attestation rewards constitute the largest part, \ie, $84.4\%$, of the consensus layer rewards~\cite{rewards}.
For a validator to receive the attestation reward,
the attestation should be included in the beacon chain and be correct and timely as defined below~\cite{rewards}. 

Correctness implies that the attestation should agree with the view of the block proposer that includes the attestation in its block. Let an attestation belong to slot $s$ in epoch $e$. Then, formally,

    \smallskip
    \noindent
    \textbf{The source checkpoint vote} is correct if the source checkpoint vote matches the justified checkpoint of epoch $e$ (\eg, the block root of the starting slot in epoch $e-1$ if epoch $e-1$ was justified).
    
    \smallskip
    \noindent
    \textbf{The target checkpoint vote} is correct if the source checkpoint vote is correct and the target checkpoint vote matches the target checkpoint of epoch $e$, \ie, the block root of the starting slot in epoch $e$.
    
    \smallskip
    \noindent
    \textbf{The head vote} is correct if the target checkpoint vote is correct and the head vote matches the head of the beacon chain at slot $s$, namely the block root of slot $s$~\cite{altair}.
    
The head vote correctness rule implies that if assuming block $B_{t+1}$ is the finalized block of slot $t+1$ in the beacon chain, then a head vote of slot $t$ is correct if and only if the head vote matches the parent root of block $B_{t+1}$.

Timeliness means that the attestation belonging to slot $s$ should get included in the beacon chain within a pre-defined slot range after slot $s$.
This range differs for the three votes included in the attestation. 
Let $t_1$ represent the slot of an attestation (by a validator assigned to the slot $t_1$ committee), 
and $t_2$ represent the slot at which the attestation gets included in the beacon chain. 
Then,
\begin{itemize}
    \item Head vote is timely if $t_2=t_1+1$.
    \item Source checkpoint vote is timely if $t_1+1 \le t_2 \le t_1+5$.
    \item Target checkpoint vote is timely if $t_1+1 \le t_2 \le t_1+32$~\cite{altair}.
\end{itemize}
The tightest slot range belongs to the head vote timeliness: if an attestation does not get included in the next slot’s block, the validator loses the head vote reward. 
To incentivize leaders to include attestations in their blocks, the proposing reward is made proportional to the total reward of the attestations and the total reward of the sync committee outputs included in the block.

Note that when the adversary reorgs a slot $t$ block $B_t$ with its slot $t+1$ block $B_{A}$, it can also include the slot $t-1$ attestations that were previously included by $B_{t}$, even by non-compliant attestors, in $B_A$ to obtain the inclusion reward for the slot $t-1$ source and target checkpoint votes.

\subsection{Slashing}
\label{sec:slashing}
Ethereum imposes financial punishments for certain detectable protocol violations~\cite{casperffg,slashing}.
If a validator is observed to have
\begin{itemize}
    \item proposed two different blocks for the same slot, or
    \item sent two or more head votes with the same slot number for distinct blocks, or
    \item sent two pairs of votes for source and target checkpoints such that one pair \emph{surrounds} the other,
\end{itemize}
then, the validator is \emph{slashed}: it is removed from the validator set, and its stake is partially burned.
Slashing rules are typically used to enforce a notion of safety called \emph{accountable safety}~\cite{casperffg,snapchat,forensics}:
\begin{definition}
\label{def:accountable-safety}
A consensus protocol satisfies accountable safety with resilience $f$, if when there is a safety violation, (i) at least $f$ adversarial validators can be identified as protocol violators with the help of a cryptographic proof, and (ii) no honest validator is identified except with negligible probability.
\end{definition}
Examples of consensus protocols with accountable safety include Ethereum, PBFT~\cite{pbft}, Tendermint~\cite{tendermint,tendermint-accountability} and HotStuff~\cite{hotstuff}.
They remain safe if there are fewer than $f$ adversarial validators (as in traditional safety); since then it would not be possible to detect $f$ adversarial validators.
These protocols come with a set of slashing rules as in Ethereum; such that if a validator violates them, then it is identified as a protocol violator and slashed.
In the following sections, we assume that the rational (as well as honest) validators never satisfy any of the slashing rules of their protocols.

\subsection{Reorgs}
\label{sec:appendix-reorgs}

Suppose an Ethereum block $B$ proposed at some time $t'$ is no longer part of the canonical chain for the first time at some time $t \geq t'$.
Then, we say that the block $B$ is \emph{reorged} at time $t$.
If $B$ proposed at some time $t'$ never becomes part of the canonical chain, it is said to be reorged at time $t'$.
Reorging of a single honestly proposed block need not imply a liveness violation.
For instance, Nakamoto consensus~\cite{bitcoin,backbone} is prone to reorgs as shown by the selfish mining attack~\cite{selfish_mining}.
Yet, in every sufficiently large period, an honest block enters and stays in the longest chain indefinitely implying liveness under an honest majority.
Conversely, any attack that can reorg all blocks proposed during an arbitrarily long period $T$ results in a liveness violation.
If these blocks are reorged after the period $T$, 
the attack would also violate safety under any reasonable confirmation rule.
Note that this does not imply a safety violation for the Ethereum blocks \emph{finalized} by Casper FFG. However, the safety of LMD GHOST is necessary for the stability of the chain and liveness of Casper FFG.

Ethereum is said to satisfy reorg resilience if an honest or rational leader's block enters and stays in the canonical chain at all times.
%
Reorg resilience implies safety and liveness of the protocol assuming that there are frequent honest or rational leaders (as is the case in RANDAO) and their blocks are not empty.

\section{Attack Quantification}\label{sec:appendix_attack_quantification}

In this section, we provide an estimate of the average reward that the adversary receives via inclusion rewards and MEV capture upon a successful attack.
We also quantify the cost of an unsuccessful attack due to the exclusion of non-compliant votes.

We first analyze the attestation inclusion reward. Each attestation consists of three different votes: a head vote, a source vote, and a target vote. If there is no attack, a slot $t+1$ block collects the inclusion reward for all three votes from a slot $t$ attestation. 
Given the number of active Ethereum validators $N = 1,073,375$ (as of October 20, 2024~\cite{ethereum_info}) and the reward distribution formula presented in~\cite{rewards}, the average attestation inclusion reward for a block, in the case of no attack, is equal to $0.0446 \, \text{ETH}$ (107 EUR\footnote{As of 23 October 2024, the price of 1 ETH is approximately 2400 EUR~\cite{eth_price}.}). 
If, under the simple attack scenario, the slot $t$ block $B_t$ is reorged, the adversary's block $B_A$ in slot $t+1$ can include the attestations from both slots $t-1$ and $t$. 
In this case, $B_A$ can collect the inclusion reward for all three votes of a slot $t$ attestation. However, it can only collect the inclusion reward for the source and target votes of a slot $t-1$ attestation, as the head vote from slot $t-1$ would not be considered timely if included in the slot $t+1$ block $B_A$. 
Therefore, the average attestation inclusion reward for an adversarial block under a successful simple attack, where all slot $t$ votes are compliant, can increase to $0.0777 \, \text{ETH}$ (186.4 EUR).

Besides the attestation inclusion reward, the execution layer reward can also increase under a successful attack. 
Figure~\ref{fig: MEV_reward_comparison} depicts the average MEV boost (execution layer reward) over the period from October 2023 to September 2024 for blocks whose previous block was either reorged or missing, compared to blocks whose parent block was proposed in the previous slot (data obtained from~\cite{slot_missed,ethereum_info}).
The former captures the reward of an adversarial block under a successful attack, while the latter gives an estimate for the adversarial block reward under a failed attack. 
As can be seen in Figure~\ref{fig: MEV_reward_comparison}, a successful simple attack can increase the MEV reward of a block from an average value of $0.082 \, \text{ETH}$ (196 EUR) to an average value of $0.12 \, \text{ETH}$ (288 EUR).
\begin{figure}[!t]
    \centering
    \centering
    \includegraphics[height=1.5in]{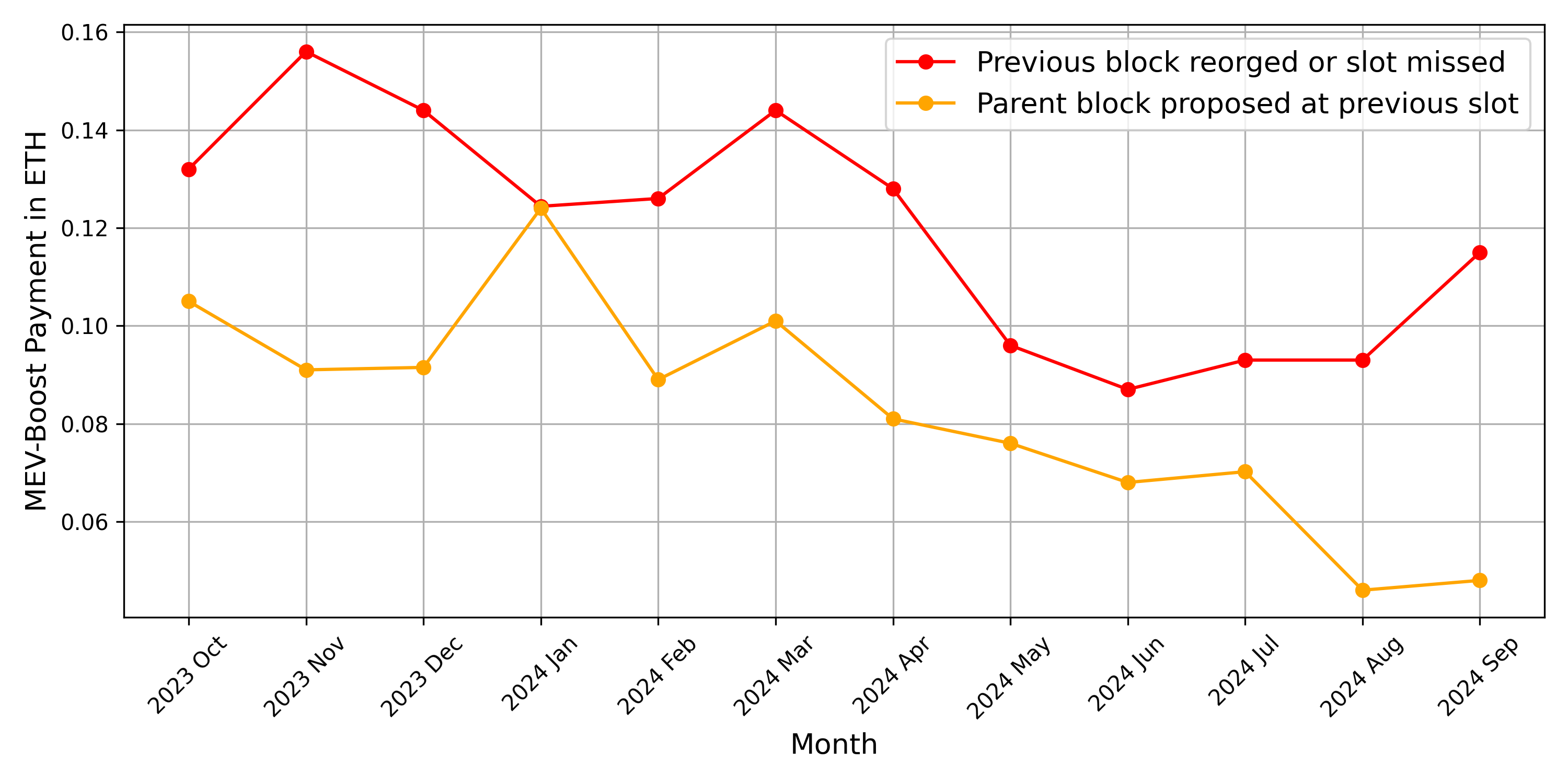}
    \caption{MEV reward comparison under a successful and failed simple attack.}
    \label{fig: MEV_reward_comparison}
\end{figure}
Therefore, considering both the inclusion and MEV rewards, the simple attack can increase the block reward from an average value of $0.1246 \, \text{ETH}$ to an average value of $0.1977 \, \text{ETH}$, representing an increase of $0.0711 \, \text{ETH}$ ($56\%$) in the block reward (170.6 EUR). Note that if a staking pool manages to reorg the slot $t$ block, it will lose the head vote rewards of those slot $t-1$ attestors that are under the control of the staking pool. Considering the largest staking pool\footnote{As of October 23, 2024, the largest staking pool is LIDO, holding $27.8$ of stake shares.}, the maximum amount of head vote rewards it loses is equal to $0.0225 \, \text{ETH}$. Therefore, by reorging a single block, a staking pool can obtain an average extra reward of at least $0.0486 \, \text{ETH}$ (116.6 EUR).

If the simple attack fails and the slot $t$ block is not reorged, $B_A$ in slot $t+1$ cannot include any attestations from slot $t-1$ and should also exclude the non-compliant attestations from slot $t$. 
If all slot $t$ attestations are non-compliant, the adversary loses, on average, the attestation inclusion reward of $0.0446 \, \text{ETH}$, but it can still get the MEV reward of a block whose parent block was proposed in time.
Note that 
if the adversary starts the simple attack on a slot $t$, where both slots $t+1$ and $t+2$ are adversarial, it can include the non-compliant attestations of slot $t$ in its slot $t+2$ block. 
In this case, the adversary can still induce a credible threat for slot $t$ attestors while only losing the inclusion reward of slot $t$ head votes, which is, on average, equal to $0.0115 \, \text{ETH}$ (27.6 EUR).

\section{Proof of Theorem~\ref{thm:repeated-game-theorem}}
\label{sec:appendix-proof-extended-game}

\begin{proof}[Proof of Theorem~\ref{thm:repeated-game-theorem}]
Suppose all solo leaders and attestors follow the adversary's game rule and propose and vote for compliant blocks in slots $1, \ldots, p$.
We show that this is a subgame perfect Nash equilibrium by backward induction.
The extended game consists of subgames played by the leaders and attestors of the slots $i \in [p]$, denoted by $\texttt{G}^L_i$ and $\texttt{G}^\val_i$ respectively.

\begin{figure*}[h]
    \centering
    \begin{subfigure}[t]{0.48\textwidth}
        \centering
        \includegraphics[height=1in]{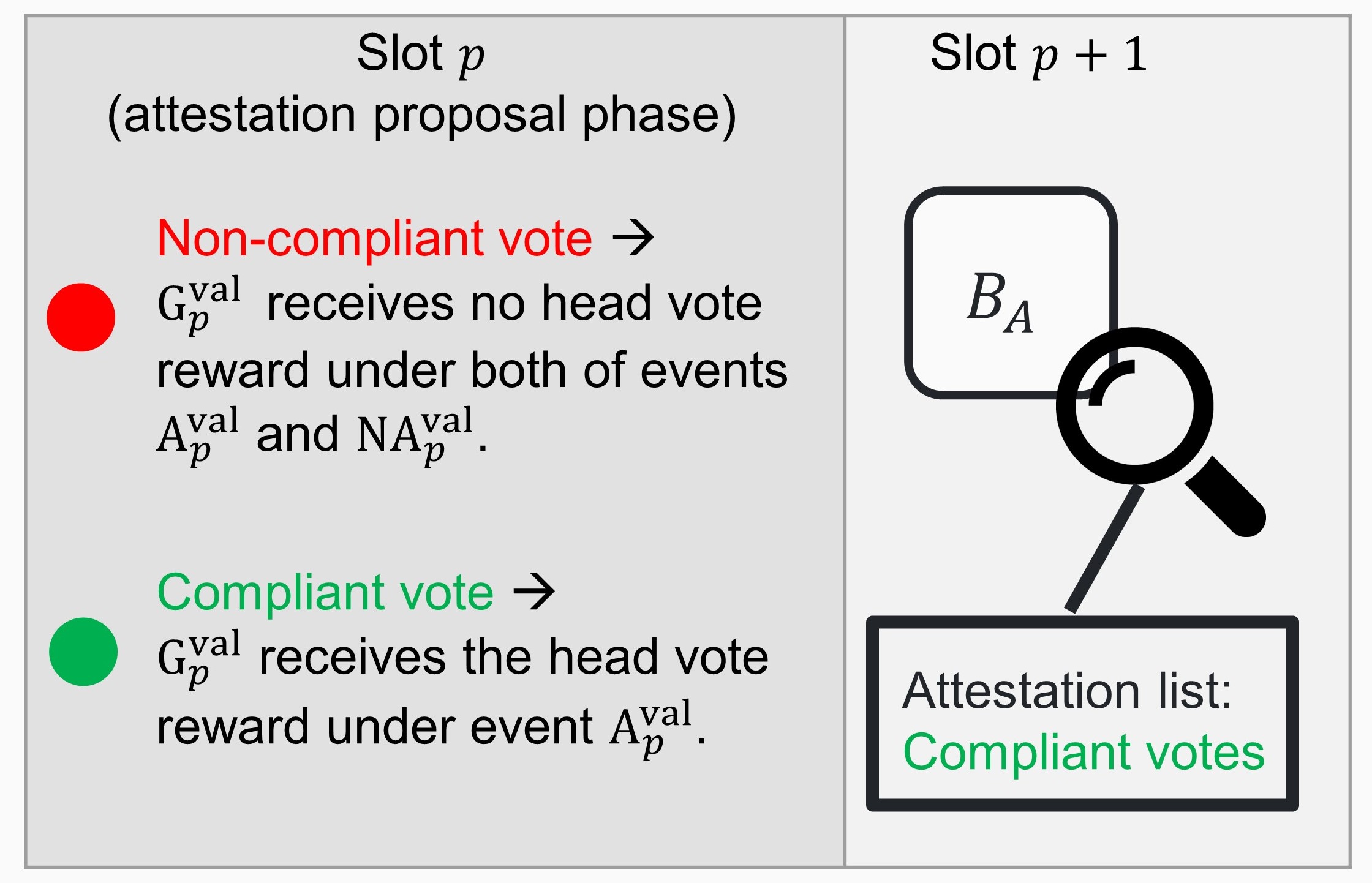}
        \caption{Subgame $\texttt{G}^\val_p$: $C^\val_p$ or $NC^\val_p$?}
        \label{fig:extended_game_subgame1}
    \end{subfigure}%
    \hfill
    \begin{subfigure}[t]{0.48\textwidth}
        \centering
        \includegraphics[height=1in]{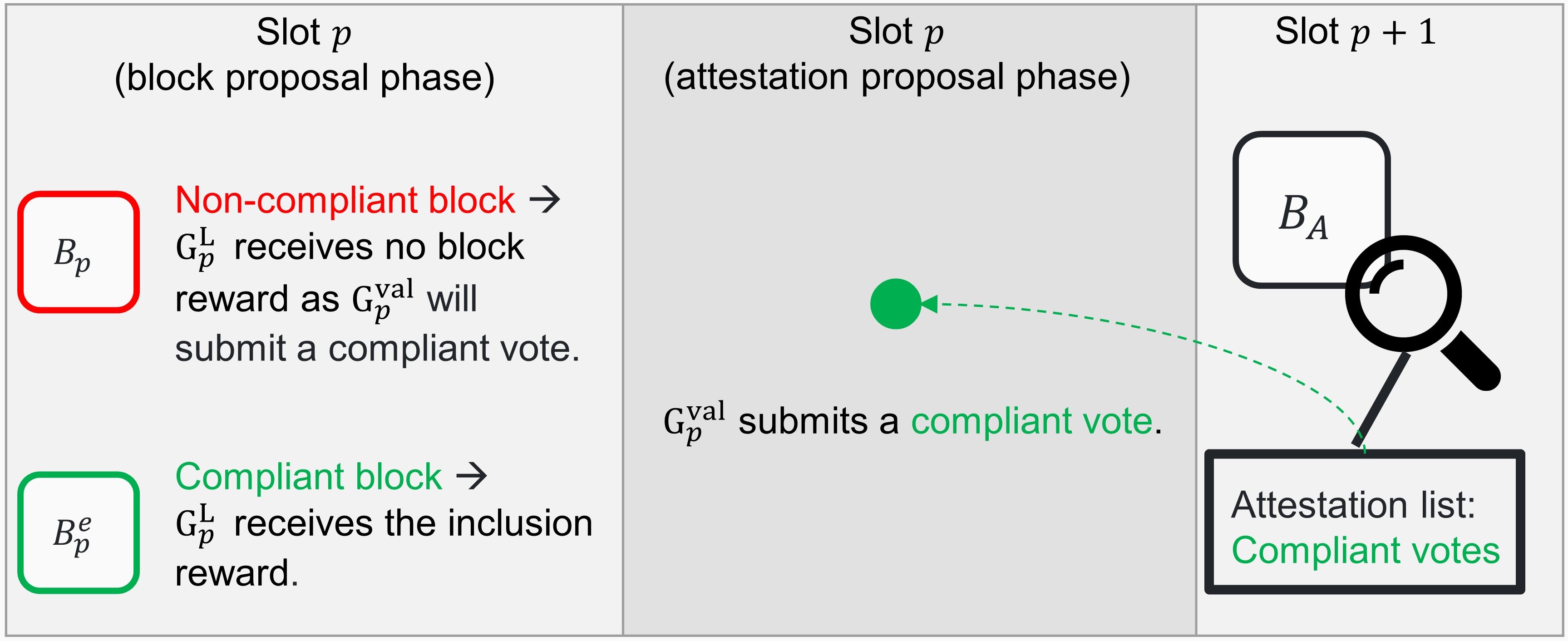}
        \caption{Subgame $\texttt{G}^L_p$: $C^L_p$ or $NC^L_p$?}
        \label{fig:extended_game_subgame2}
    \end{subfigure}
    
    \vskip\baselineskip 
    
    \begin{subfigure}[t]{0.48\textwidth}
        \centering
        \includegraphics[height=1in]{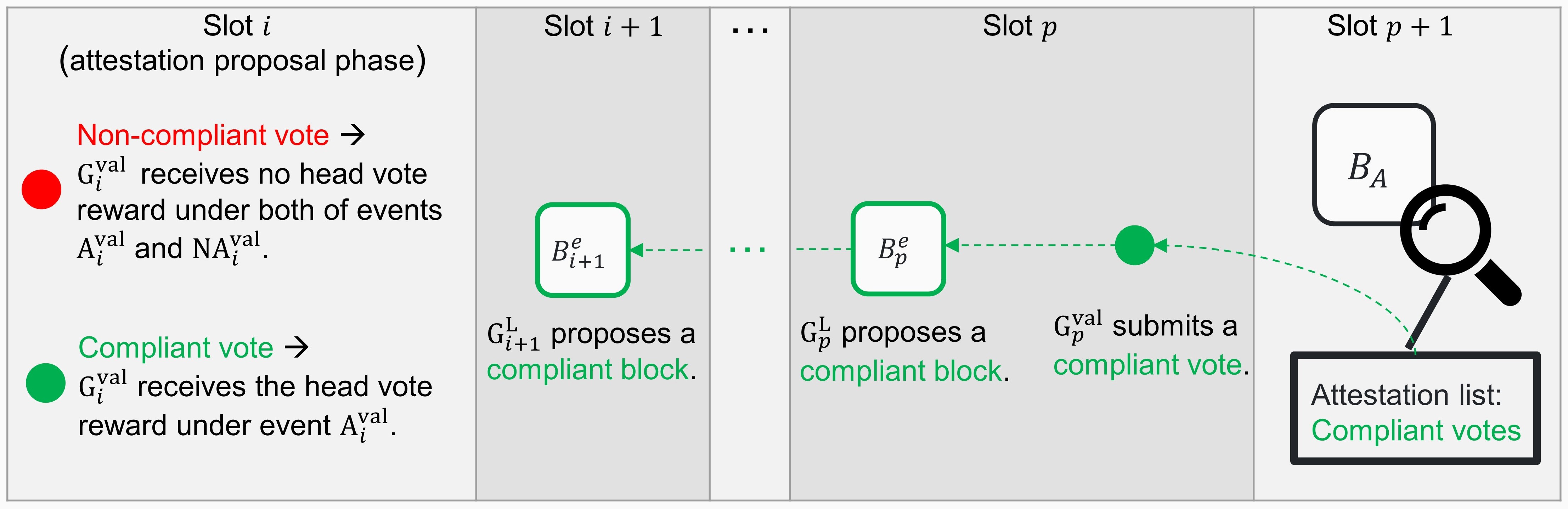}
        \caption{Subgame $\texttt{G}^\val_i$: $C^\val_i$ or $NC^\val_i$?}
        \label{fig:extended_game_subgame_2(p-i)+1}
    \end{subfigure}%
    \hfill
    \hfill
    \begin{subfigure}[t]{0.48\textwidth}
        \centering
        \includegraphics[height=0.9in]{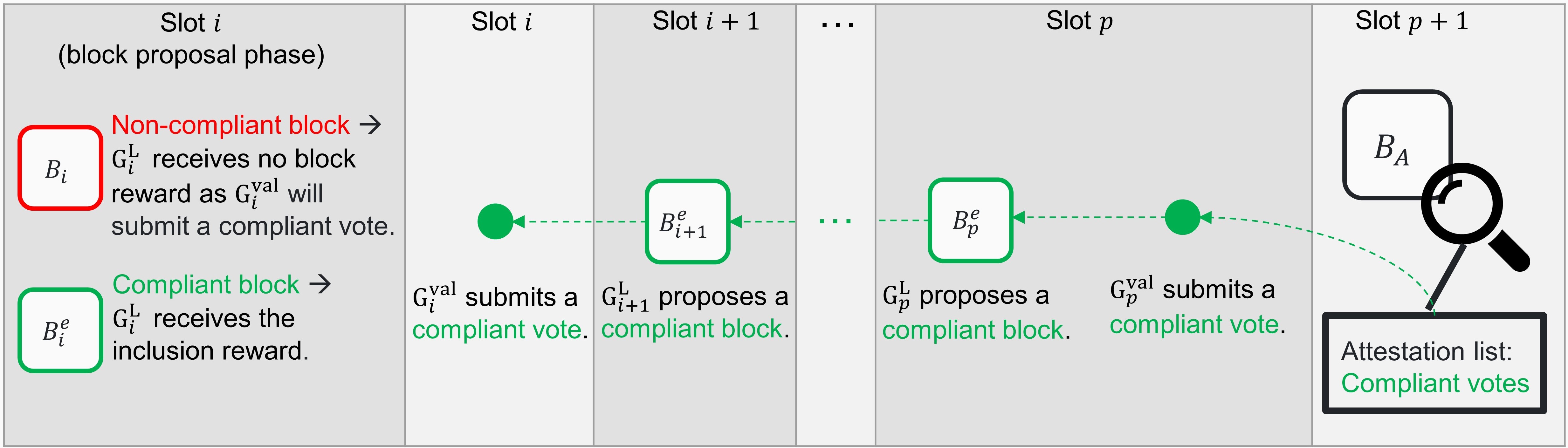}
        \caption{Subgame $\texttt{G}^L_i$: $C^L_i$ or $NC^L_i$?}
        \label{fig:extended_game_subgame_2(p-i)+2}
    \end{subfigure}
    
    \caption{Subgame perfect Nash equilibrium}
\end{figure*}

The payoff matrix of a solo slot $p$ attestor $\val$ in $\texttt{G}^\val_p$ is presented in Table~\ref{table:payoff_matrix-repeated_game_p} (Fig.~\ref{fig:extended_game_subgame1}).
\begin{table}[t]
    \centering
    \caption{Payoff matrix of a solo slot $p$ validator $\val$ in the extended game. 
    Let $C^\val_p$ denote the event that $\val$ sends a compliant vote.
    Let $A^\val_p$ denote the event that all slot $p$ attestors send compliant votes.}
    \begin{tabular}{|c|c|c|}
        \hline
          & $C^\val_p$ & $NC^\val_p$ \\
         \hline
         $A^\val_p$ & $r$ & $0$ \\
         \hline
         $NA^\val_p$ & $\leq r$ & $0$\\
         \hline
    \end{tabular}
    \label{table:payoff_matrix-repeated_game_p}
\end{table}
A non-compliant vote is never included in $B_A$ and cannot be timely, which implies a payoff of $0$ ($NC^\val_p$ column).
If $\val$ sends a compliant vote for some block $B^*$ ($C^\val_p$ column) and all slot $p$ attestors send compliant votes\footnote{In practice, the adversary can use the $\Wp$ boost of its block. Moreover, the blocks $\mathcal{A}$ tries to reorg will not all have the maximum $W$ votes, implying that not all of the $W$ slot $p$ votes will be needed.}, which must be for $B^*$ as ties are broken deterministically ($A^\val_p$ row), then upon running Alg.~\ref{alg.compliant}, $\mathcal{A}$ obtains the same block $B^*$ as its compliant tip and extends it with its block $B_A$.
In this case, $\val$ gains a payoff of $r$ ($A^\val_p$-$C^\val_p$ corner).
Hence, for $\val$, sending a compliant vote weakly dominates any other action, and it follows the game rule.

Suppose the slot $p$ attestors send compliant votes.
Then, the slot $p$ leader $L$'s payoff matrix in $\texttt{G}^L_p$ is presented in Table~\ref{table:payoff_matrix-repeated_game_p_leader} (Fig.~\ref{fig:extended_game_subgame2}).
Let $R$ represent the inclusion reward that a slot leader receives for proposing a block in the canonical chain.
If $L$ proposes a compliant block $B^e_p$, all slot $p$ attestors vote for it.
Then, the compliant tip obtained by $\mathcal{A}$ includes $B^e_p$ in its prefix, which implies a payoff of $R$ for $L$ ($C^L_p$ corner).
However, if $L$ proposes a non-compliant block, the block does not gather any slot $p$ votes.
In this case, the compliant tip obtained by $\mathcal{A}$ does not include $L$'s block in its prefix, implying a payoff of $0$ for $L$ ($NC^L_p$ corner).
Thus, proposing a compliant block $B^e_p$ weakly dominates proposing any other block, and $L$ follows the adversary's game rule.

Suppose all leaders and attestors of the slots $i+1, \ldots, p$ follow the adversary's game rule, \ie, cast compliant votes and propose compliant blocks.
Then, the payoff matrix of a solo slot $i$ attestor $\val$ in $\texttt{G}^\val_{i}$ is presented in Table~\ref{table:payoff_matrix-repeated_game_i} (Fig.~\ref{fig:extended_game_subgame_2(p-i)+1}).
Since the slot $i+1$ block is compliant, a non-compliant vote would not be included in the block of slot $i+1$, and thus cannot be timely, which implies a payoff of $0$ ($NC^\val_i$ column).

When all leaders and attestors of the slots $i+1, \ldots, p$ follow the adversary's game rule, if $\val$ sends a compliant vote for some block $B^*$, ($C^\val_i$ row) and all slot $i$ attestors send compliant votes (for $B^*$, $A^\val_i$ row), all of the blocks $B^e_{j}$, $j \in \{i+1,\ldots,p\}$ and $B_A$ extend $B^*$, which implies a payoff of $r$ ($C^\val_i$-$A^\val_i$ corner).
Hence, for $\val$, sending a compliant vote weakly dominates any other action, and it follows the game rule.
\begin{table}[b]
    \centering
    \caption{Payoff matrix of the slot $p$ leader $L$ in the extended game. Let $C^L_p$ denote the event that $L$ proposes a compliant block.
    }
    \begin{tabular}{|c|c|}
        \hline
         $C^L_p$ & $NC^L_p$ \\
         \hline
         $R$ & $0$ \\
         \hline
    \end{tabular}
    \label{table:payoff_matrix-repeated_game_p_leader}
\end{table}
\begin{table}[t]
    \centering
    \caption{Payoff matrix of a solo slot $i$ validator $\val$ in the extended game. 
    Let $C^\val_i$ denote the event that $\val$ sends a compliant vote.
    Let $A^\val_i$ denote the event that all slot $i$ attestors send compliant votes.}
    \begin{tabular}{|c|c|c|}
        \hline
          & $C^\val_i$ & $NC^\val_i$ \\
         \hline
         $A^\val_i$ & $r$ & $0$ \\
         \hline
         $NA^\val_i$ & $\leq r$ & $0$\\
         \hline
    \end{tabular}
    \label{table:payoff_matrix-repeated_game_i}
\end{table}

Suppose all leaders of the slots $i+1, \ldots, p$, and all attestors of the slots $i, \ldots, p$, follow the adversary's game rule, \ie, cast compliant votes and propose compliant blocks.
Then, the slot $i$ leader $L$'s payoff matrix in $\texttt{G}^L_i$ is presented in Table~\ref{table:payoff_matrix-repeated_game_i_leader} (Fig.~\ref{fig:extended_game_subgame_2(p-i)+2}).
If $L$ proposes a compliant block $B^e_i$, all slot $i$ attestors vote for it.
When all leaders and attestors of the slots $i, i+1, \ldots, p$ follow the adversary's game rule, all of the blocks $B^e_{j}$, $j \in \{i+1,\ldots,p\}$ and $B_A$ extend $B^e_i$, which implies a payoff of $R$ ($C^L_i$ corner).
However, if $L$ proposes a non-compliant block, the block does not gather any slot $p$ votes.
In this case, the compliant tip obtained by any subsequent leader, including $\mathcal{A}$, does not include $L$'s block in its prefix by Alg.~\ref{alg.compliant}, implying a payoff of $0$ for $L$ ($NC^L_i$ corner).
Thus, proposing a compliant block $B^e_i$ weakly dominates proposing any other block, and $L$ follows the adversary's game rule.
\begin{table}[b]
    \centering
    \caption{Payoff matrix of the slot $i$ leader $L$ in the extended game. Let $C^L_i$ denote the event that $L$ proposes a compliant block.
    }
    \begin{tabular}{|c|c|}
        \hline
         $C^L_i$ & $NC^L_i$ \\
         \hline
         $R$ & $0$ \\
         \hline
    \end{tabular}
    \label{table:payoff_matrix-repeated_game_i_leader}
\end{table}

Therefore, there exists a subgame perfect Nash equilibrium, where all solo leaders and attestors follow the adversary's game.
Then, the chain ending at the sequence of blocks $B^e_1, \ldots, B^e_p, B_A$ becomes the canonical chain after the end of slot $p$, and the adversary wins the game.
In this case, safety and liveness are violated since the blocks $B_{-p+1}, \ldots, B_{0}$ are reorged.
\end{proof}

\begin{note}
\label{note:5.1}
There are many subgame perfect Nash equilibria, where the adversary $\mathcal{A}$ is not successful. 
Suppose the slot $i \in [p]$ leaders propose non-empty blocks $B_i$ extending the blocks $B_0, \ldots, B_{i-1}$, and over $1+W/2$ slot $i$ attestors vote for $B_{i}$.
Then, $\mathcal{A}$ is unsuccessful, and no slot $i$ vote for a block other than $B_i$ can be both correct and timely, whereas due to the adversary's commitment, no slot $p$ vote can receive a positive payoff for any action.
Thus, at each subgame, 
voting for $B_i$ weakly dominates any other action for the slot $i$ committee, and 
proposing a non-empty block $B_i$ extending $B_{i-1}$ weakly dominates any other action for the slot $i$ leader.
Therefore, the situation above is a SPNE.
\end{note}


\begin{note}
\label{note:5.3}
Under an honest minority assumption (\ie, number of honest validators per slot $W_h<(W/2)$), LMD GHOST is still susceptible to the extended attack, \ie, there might be a SPNE where the adversary succeeds.
In the extended game, the adversary can run the game for $p'=pW/(W-2W_h)$ blocks, gathering enough votes on a conflicting chain of $p'$ blocks from slots $1, \ldots, p'$ to reorg the $p$ blocks from slots $-p+1, \ldots, 0$.
\end{note}

\section{The Simple Attack in the Presence of Staking Pools}
\label{sec:appendix-proof-simple-game-fixed}


\begin{proof}[Proof of Theorem~\ref{thm:fixed-validator-set-simple-game-theorem}]
To analyze the simple game in the presence of staking pools, we calculate the payoff matrix of a staking pool $P$. 
Suppose the staking pool $P$ has $m < W_p$ attestors in each slot, and recall that $r$ represents the reward of an attestor for a correct and timely vote included in the canonical chain.
Assume further that the slot $t$ leader is not under $P$'s control.
The pay-off that a staking pool receives in the simple game is equal to the sum of the vote rewards that all its $2m$ attestors receive in both slots $t-1$ and $t$. 
The payoff matrix of staking pool $P$ is presented in Table~\ref{table:pay-off-matrix_simple-game_fixed-validator-set}. 
Each entry in the table is calculated as the sum of two components: The left component represents the vote reward of slot $t-1$, while the right component represents the vote reward of slot $t$.
Overall, 4 different scenarios can occur:
    
    \smallskip
    \noindent
    $\mathbf{\texttt{Fail}-NC_P}$ \textbf{corner (Fig.~\ref{fig:simple_game_staking_valool_AV}):} $\Wp$ or more slot $t$ attestors and staking pool $P$ validators vote for $B_t$.
    In this case, $B_t$ cannot be reorged by $B_{t+1}$, and $P$'s votes for slot $t-1$ are included in $B_t$.
    Thus, $P$ receives a payoff of $mr$ for slot $t-1$. 
    However, since the $P$'s votes for slot $t$ are non-compliant, 
    the adversary $\mathcal{A}$ excludes $P$'s votes from its slot $t+1$ block $B_A$. 
    Therefore, $P$'s reward for slot $t$ is equal to $0$.
    
    \smallskip
    \noindent
    $\mathbf{\texttt{Fail}-C_P}$ \textbf{corner (Fig.~\ref{fig:simple_game__staking_valool_ANV}):} $\Wp$ or more slot $t$ attestors vote for $B_t$, but the staking pool $P$ validators do not vote for $B_t$, instead voting for $B_t$'s parent.
    In this case, $B_t$ does not get reorged, and $P$ receives a payoff of $mr$ for slot $t-1$. 
    As over $\Wp$ slot $t$ attestors have voted for $B_t$, $\mathcal{A}$ proposes its block on top of $B_t$.
    Therefore, $P$ does not receive any payoff for its slot $t$ votes for $B_{t-1}$, as they are not correct. 
    
    \smallskip
    \noindent
    $\mathbf{\texttt{Succeed}-NC_P}$ \textbf{corner (Fig.~\ref{fig: simple_game__staking_valool_NAV}):} $\Wp$ or less slot $t$ attestors vote for $B_t$, but $P$ votes for $B_t$.
    Then, the attack is successful, and $B_t$ gets reorged. As a result, $P$'s slot $t-1$ vote reward, \ie, payoff, is equal to $0$.
    Since $P$'s slot $t$ votes are non-compliant, its payoff for slot $t$ is also equal to $0$.
    
    \smallskip
    \noindent
    $\mathbf{\texttt{Succeed}-C_P}$ \textbf{corner (Fig.~\ref{fig:simple_game_staking_valool_NANV}):} $\Wp$ or less slot $t$ attestors vote for $B_t$, and $P$ does not vote for $B_t$.
    Then, the attack is successful, and $\mathcal{A}$ reorgs $B_t$ with its block. 
    Since block $B_t$ gets reorged, $P$'s payoff for slot $t-1$ is equal to $0$, but it receives a payoff of $mr$ for its compliant slot $t$ votes.
%
\begin{table}[t]
    \centering
    \caption{Payoff matrix of a staking pool $P$ in the simple game. Let $C_P$ denote the event that the staking pool $P$'s validators vote for $B_{t-1}$ (\ie, $P$ is compliant). Let \texttt{Succeed} denote the event that less than $\Wp$ slot $t$ attestors vote for $B_t$ (attack succeeds), and let \texttt{Fail} denote the event that $\Wp$ or more slot $t$ attestors vote for $B_t$ (attack fails).}
    \begin{tabular}{|c|c|c|}
        \hline
          & $C_P$ & $NC_P$ \\
         \hline
         $\texttt{Succeed}$ & $0+mr$ & $0+0$ \\
         \hline
         $\texttt{Fail}$ & $mr+0$ & $mr+0$\\
         \hline
    \end{tabular}
    \label{table:pay-off-matrix_simple-game_fixed-validator-set}
\end{table}

As can be seen in Table~\ref{table:pay-off-matrix_simple-game_fixed-validator-set}, similar to the solo validator analysis, voting for the parent block $B_{t-1}$ is a weakly dominant strategy for a staking pool. This indicates that if rational validators constitute over $W-\Wp$ of the committees, the simple attack would be successful even in the presence of staking pools. 
\end{proof}

\section{The Extended Attack in the Presence of Staking Pools}
\label{sec:appendix-proof-extended-game-fixed}


\begin{proof}[Proof of Theorem~\ref{thm:fixed-validator-set-extended-game-theorem}]
Consider a staking pool $P$ that controls $m<W_p$ attestors in each slot, and recall that $r$ represents the reward of an attestor for a correct and timely vote included in the canonical chain.
Suppose all slot $1, \ldots, p$ leaders propose compliant blocks and all slot $1, \ldots, p$ attestors vote for compliant blocks, resulting in the reorg of the blocks from slots $-p+1, \ldots, 0$. 
We show that this is a subgame perfect Nash equilibrium by backward induction.

Consider a branch of actions taken by the players, in which slot $1, \ldots, p$ leaders propose compliant blocks and all slot $1, \ldots, p$ attestors except those in $P$ vote for compliant blocks (those in $P$ might or might not have voted for compliant blocks).
Let $B^e_p$ denote the compliant block voted by the slot $p$ attestors.
Then upon running Alg.~\ref{alg.compliant}, $\mathcal{A}$ obtains the same block $B^e_p$ as its compliant tip (as $W_s < W_p$) and extends it with its block $B_A$.
Now, if $P$ sends compliant votes, which must be for $B^e_p$ as ties are broken deterministically, then its votes are included by $B_A$, and it achieves a total payoff of $mr$ for slot $p$.
If $P$ does not send compliant votes, its total payoff for slot $p$ becomes $0$ as $B_A$ would not include its votes.
Hence, $P$ also sends compliant votes for block $B^e_p$.

Since slot $p$ attestors send compliant votes, the slot $p$ leader $L$'s block becomes part of the canonical chain only if it proposes a compliant block $B^e_p$, in which case it attains a reward of $R$.
Otherwise, its reward is $0$.
Thus, it proposes a compliant block, \ie, $B^e_p$.

Next, consider a branch of actions taken by the players, in which slot $1, \ldots, i$ leaders propose compliant blocks and all slot $1, \ldots, i$ attestors except those in $P$ vote for compliant blocks (those in $P$ might or might not have voted for compliant blocks).
In this case, let $B^e_i$ denote the compliant block voted upon by the slot $i$ attestors.
Then, the canonical chain at the start of slot $p+1$ contains $B^e_i$ and all the compliant blocks extending $B^e_i$ by backward induction.
Now, if $P$ sends compliant votes, which must be for $B^e_i$ as ties are broken deterministically, then the votes are included by the compliant block $B^e_{i+1}$ extending $B^e_i$ in the canonical chain, and $P$ achieves a total payoff of $mr(p-i)$ for slots $i, \ldots, p$ by backward induction.
If $P$ does not send compliant votes, its total payoff for slots $i, \ldots, p$ becomes $mr(p-i-1)$ as $B^e_{i+1}$ would not include the votes.
Hence, $P$ also sends compliant votes for block $B^e_i$.

Since slot $i$ attestors send compliant votes, the slot $i$ leader $L$'s block becomes part of the canonical chain only if it proposes a compliant block $B^e_i$, in which case it attains a reward of $R$.
Otherwise, its reward is $0$.
Thus, it proposes a compliant block, \ie, $B^e_i$.
This concludes the proof.
\end{proof}

\begin{note}
\label{note:extended_game_fixed_validator}
There are SPNE, similar to those in Note~\ref{note:5}, where the adversary $\mathcal{A}$ is not successful.
\end{note}

\section{Performance Comparison: The Optimistic Scenario}
\label{sec:appendix-optimistic}

In this section, we estimate the storage, computation, and communication overheads of our solution and compare them with the overheads of the current Ethereum version in an \emph{optimistic scenario}, where the majority of slot $t+1$ aggregators sign the same slot $t$ aggregate.

\textbf{Block space}: As of October 23, 2024, each block of Ethereum contains 128 aggregated attestation signatures as well as 128 aggregation lists, where the size of each signature and each aggregation list is equal to 96 bytes and $N^\text{att}=524$ bits, respectively. Therefore, the total size of
    aggregated signatures and aggregation lists (aggregates)
    in a single block is 20.672 kilobytes. 
    The average Ethereum block size from January 1, 2024, to October 23, 2024, is approximately $101.5$ kilobytes~\cite{Block_size}.
    
    Under the optimistic scenario, each block includes $128$ aggregated evidences for an earlier slot (one per sub-committee in the earlier slot). Each piece of aggregated evidence comprises 4 parts: the aggregated evidence signature ($96$ bytes), the evidence aggregation list ($N^\text{agg}$ bits), the aggregated attestation signature ($96$ bytes), and the attestation aggregation list ($N^\text{att}$ bits). 
    Assuming that $N^\text{att}=524$, the total size of aggregated evidences in each block is equal to $32.96+ 0.016*N^\text{agg}$ kilobytes. 
    Suppose the parameters $N^\text{agg}$ and $N^\text{limit}$ are set such that the probability that the number of non-adversarial aggregators in a sub-committee becomes less than or equal to $N^\text{limit}$ is less than $10^{-4}$.
    \begin{itemize}[leftmargin=*]
        \item Assuming the adversarial stake share is less than or equal to $0.1$, $N^\text{agg}$ and $N^\text{limit}$ can be set to 16 and 8, respectively. In this case, the total size of evidences in a single block would be equal to $33.216$ kilobytes. As these evidences replace the 20.672 kilobyte aggregates in the current Ethereum version, the block size increases by 12.544 kilobytes. This results in a 12.3\% increase in the block size.
        \item Assuming the adversarial stake share is less than or equal to $1/3$, $N^\text{agg}$ and $N^\text{limit}$ can be set to 128 and 64, respectively. In this case, the total size of evidences per block would be equal to $35.008$ kilobytes, which is 14.336 kilobytes higher than the current size of block aggregates. This results in an 14.1\% increase in the block size.
    \end{itemize}

\noindent
 \textbf{Computational overhead}: We assume $N^\text{att}=524$ and $N^\text{agg}=16$. Let $C^\text{add}$, $C^\text{mul}$, and $C^\text{pair}$ denote the computational cost of an addition, a scalar multiplication, and a pairing operation over the elliptic curve, respectively.
    \smallskip
    \noindent
    \emph{Aggregators:} In the current Ethereum version, an aggregator must first verify the $N^\text{att}$ single attestation signatures of the corresponding sub-committee, where each signature verification requires $2$ pairing operations. 
    Then, the aggregator has to perform $N^\text{att}-1$ additions over the elliptic curve to aggregate attestations and $1$ scalar multiplication over the elliptic curve to sign the aggregated attestation. Therefore, its computational cost becomes $2N^\text{att}C^\text{pair}+(N^\text{att}-1)C^\text{add} + C^\text{mul}$. 
    In the practical version of our solution, in addition to these operations, a slot $t$ aggregator must verify and sign $1$ aggregated attestation (the one that includes the greatest number of attestations) from slot $t-1$ to generate an evidence. 
    The aggregated attestation can be verified by checking its aggregated signature that contains the signatures of the attestors in the corresponding sub-committee from slot $t-1$.
     To check this signature combining $N^\text{att}$ individual signatures, a verifier needs to perform $N^\text{att}-1$ additions over the elliptic curve to calculate the aggregated public key of the participating attestors, along with 2 pairing operations. Thus, an aggregator's computational cost in our practical version becomes $2(N^\text{att}+1) C^\text{pair} + 2(N^\text{att}-1)C^\text{add}+ 2C^\text{mul}$. Based on our implementation results\footnote{These results were achieved through implementation on a laptop equipped with an Intel Core i7 CPU operating at a processor base frequency of 2.70GHz.}, this computation takes 0.699 seconds.
        
    \smallskip
    \noindent
    \emph{Block proposers:} In our practical version, each block proposer needs to verify and aggregate $N^\text{agg}$ evidences per each sub-committee to generate an aggregated evidence, in addition to standard operations of the current Ethereum. The computational costs of verification and aggregation of evidences per sub-committee are equal to $2N^\text{agg}C^\text{pair}$ and $(N^\text{agg}-1)C^\text{add}$, respectively. As the number of sub-committees per slot is $64$, a block proposer in total incurs an additional computational cost of $128N^\text{agg}C^\text{pair}+64(N^\text{agg}-1)C^\text{add}$. Based on our implementation results, this additional computation takes $1.36$ seconds.
    
    \smallskip
    \noindent
    \emph{Verifiers:} To verify the consensus layer-related data of a block in the current Ethereum version, a verifier must check $64$ aggregated signatures, each belonging to one of the sub-committees. This results in a total computational cost of $64(N^\text{att}-1)C^\text{add} + 128C^\text{pair}$. In our practical version, in addition to this verification, a verifier must verify $1$ aggregated evidence per each sub-committee, where each aggregated evidence contains $N^\text{agg}$ evidences. All the evidences of a sub-committe are created over the same aggregated attestation that includes $N^\text{att}$ attestations. This implies that the evidence verification cost per sub-committee is equal to $(N^\text{att}+N^\text{aagg}-2)C^\text{add} + 4C^\text{pair}$. Therefore, and a verifier incurs the computational cost of $64(2N^\text{att}+N^\text{agg}-3)C^\text{add} + 384C^\text{pair}$ in the practical version. Based on our implementation, this computation takes $0.402$ seconds.
    %

    \noindent\textbf{Communication overhead}: Under the optimistic scenario, each aggregator needs to send an additional pair of messages: the aggregated attestation  ($96$ bytes) from the previous slot and its signature on that attestation  ($96$ bytes). This implies that each aggregator has to send $192$ extra bytes compared to the current Ethereum version.

\section{Performance Comparison: The Worst Case Scenario}
\label{sec:appendix-worst-case}

In this section, we estimate the storage, computation, and communication overheads of our solution and compare them with the overheads of the current Ethereum version in the \emph{worst-case scenario}, which happens if the slot $t+1$ block is missing or has excluded the attestations, and the slot $t+1$ aggregators sign different aggregates.
Processing these diverse evidences can impose additional overheads on the blockchain network. 
In the following, we present the analysis of the overheads incurred by our solution in the worst-case scenario:

    \smallskip
    \noindent
    \textbf{Block space}:
    In the worst-case scenario, each block includes $128*N^\text{agg}$ individual evidences ($N^\text{agg}$ evidences per each sub-committee). Each piece of individual evidence comprises 3 parts: the evidence signature ($96$ bytes), the aggregated attestation signature ($96$ bytes), and the attestation aggregation list ($N^\text{att}$ bits).
    Assuming that $N^\text{att}=524$, the total size of evidences in each block is equal to $32.96*N^\text{agg}$ kilobytes. 
    Suppose the parameters $N^\text{agg}$ and $N^\text{limit}$ are set such that the probability that the number of non-adversarial aggregators in a sub-committee becomes less than or equal to $N^\text{limit}$ is less than $10^{-4}$.
    \begin{itemize}[leftmargin=*]
        \item Assuming the adversarial stake share is less than or equal to $0.1$, $N^\text{agg}$ and $N^\text{limit}$ can be set to 16 and 8, respectively. In this case, the total size of evidences in a single block would be equal to 527.36 kilobytes. 
        \item Assuming the adversarial stake share is less than or equal to $1/3$, $N^\text{agg}$ and $N^\text{limit}$ can be set to 128 and 64, respectively. In this case, the total size of evidences in a single block would be equal to $4.218$ megabytes.
    \end{itemize}
    
    \smallskip
    \noindent
     \textbf{Computational overhead}: We assess the computational overhead under the worst-case scenario from the perspective of the following parties:
    \begin{itemize}[leftmargin=*]
        \item Aggregators: It is similar to the optimistic scenario.
        \item Block proposers: In the worst-case scenario, each block proposer needs to verify $N^\text{agg}$ evidences per each sub-committee to generate an aggregated evidence, in addition to standard operations in the current Ethereum version. The computational cost of evidence verification per sub-committee is equal to $2N^\text{agg}C^\text{pair}$. As the number of sub-committees per slot is equal to $64$, a block proposer in total incurs an additional computational cost of $128N^\text{agg}C^\text{pair}$.
        \item Verifiers: Under the worst-case scenario, in addition to the existing verification process in the current Ethereum version, a verifier needs to verify $N^\text{agg}$ single evidences per each sub-committee, where each single evidence is created over an aggregated attestation that includes $N^\text{att}$ attestations. This implies that the evidence verification cost per sub-committee is equal to $(N^\text{att}-1)N^\text{agg}C^\text{add} + 4N^\text{agg}C^\text{pair}$. Therefore, a verifier incurs the computational cost of $64(N^\text{att}-1)(N^\text{agg}+1)C^\text{add} + 64(2+4N^\text{agg})C^\text{pair}$ in the worst-case scenario.
    \end{itemize}

    
    \smallskip
    \noindent
    \textbf{Communication overhead} is the same as the optimistic case.

\section{Impact of Future Ethereum Changes}
\label{sec:appendix_discussion}
%

In this section, we analyze the impact of potential future changes in the Ethereum protocol~\cite{roadmap} on our attacks. 
None of these modifications can mitigate the attack.
The attack can even become more destructive with certain changes, such as single-slot finality.

    
\subsection{\text{[block, slot]-voting}~\cite{block_slot_voting}} 
In the current Ethereum, if no block is proposed during a slot, the slot's assigned committee votes for the previous slot's block as the head block.
In the [block, slot]-voting mechanism, the attestors of a slot with no block vote on an \emph{empty slot}. 
The [block, slot]-voting might be adopted by the Ethereum protocol in the future as it 
facilitates the implementation of an in-protocol proposer-builder-separation mechanism. 
    
A slightly modified version of our attack is still applicable to the protocol with the [block, slot]-voting mechanism. 
For instance, in the simple game of Section~\ref{sec:simple-game}, the game rule only needs to specify ``the attestors of slot $t$ to vote for the empty slot" rather than the previous block.
Similarly, the extended game would stipulate the attestors to vote for empty slots during slots in $[p]$.
    
\subsection{Secret Leader Election (SLE)~\cite{secretLeaderElection}} 
Currently, the list of upcoming block proposers in Ethereum is public, meaning that a validator can predict future slot leaders. 
SLE can ensure that only the selected slot leader knows its role as a block proposer in one of the upcoming slots. 
SLE can improve the resistance of Ethereum against the denial-of-service (DOS) attack.

SLE cannot mitigate our attack as the attacker only needs to be aware of the future time slots in which they \emph{themselves} are selected as the slot leader, which is feasible even with the adoption of SLE.
Moreover, the attack could still be modified to work under protocol changes where leaders themselves do not know their slots until the slot starts.
For instance, the adversary could deploy a contract that commits to rewarding a slot $t$ attestor, whenever the attestor can prove that the adversary is the leader of some slot $t+1$ and that the attestor satisfies the compliance conditions with its slot $t$ votes.
The contract can also be permissionlessly triggered to slash the adversary if it happens to exclude compliant votes.
    
\subsection{In-protocol Proposer--Builder Separation (ePBS)~\cite{PBS}} 
Currently, Ethereum proposers (slot leaders) are responsible for creating a bundle of transactions and including them in the proposed block. 
They may decide to outsource the responsibility of bundle building to builders operating outside of the Ethereum protocol. 
In this context, ePBS is a prominent direction (among many~\cite{mev-directions}) considered for the future of MEV allocation.
It splits the task of building a transaction bundle and proposing a block between two validators: a block builder and a block proposer. 
Block builders become responsible for creating transaction bundles and offering them to block proposers. 
Block proposers become responsible only for proposing and sending out the blocks to peers on the network without knowing the content of the included transaction bundle. 
ePBS can offer a fairer distribution of MEV opportunities~\cite{enshrine_PBS}.

Our attack can also threaten the ePBS version of Ethereum. 
A malicious slot $t$ leader can incentivize the attestors of slot $t-1$ to vote on an unrevealed transaction bundle. 
By doing so, the victim block builder of slot $t-1$, who has published his transaction bundle honestly, may get penalized, and the included MEV in the victim transaction bundle can be stolen.


\subsection{Single-Slot Finality~\cite{singleSlotFinality}} 
In present-day Ethereum, which uses Casper FFG as finality gadget, a block has to wait for 2 epochs to get finalized. 
The concept of single-slot finality aims to achieve the goal of finalizing blocks within the same slot they were proposed. 
It can thus significantly reduce the transaction confirmation time. 
    
The security of protocols with single-slot finality~\cite{DAmatoZ23} can be proven based on the honest majority assumption for online validators~\cite{simple_singleSlotFinality}. 
In our attacks, the adversary need not control any majority.
Thus, the attacks would succeed in creating equilibria, where liveness can be violated for arbitrarily long slot sequences, on Ethereum with single-slot finality as long as the majority of the validators are profit-maximizing (rational) agents.
Even though the attacks would not be able to reorg a long sequence of finalized blocks (recall that in the case of LMD GHOST, a long sequence could be reorged before the Casper FFG finalizations catch up, \cf, Section~\ref{sec:solo-validators-extended}),
if the adversary is only interested in violating liveness but not safety, 
it can induce subgame perfect equilibria, where a long sequence of empty blocks are finalized.

\section{Simple Attack without Proposer Boost}
\label{sec:appendix-simple-attack-wo-boost}
The simple attack introduced in Section~\ref{sec:simple-game} benefits from a high proposer boost. However, even reducing the proposer boost to $0$ cannot completely mitigate this attack. In this section, we introduce a modified version of the simple attack that works under the assumption of $\Wp = 0$.

Assume $\Wp=0$. The attack applies to some slot $t$ where slots $t-1$ and $t+1$ are adversarial. Let $B_0$ denote the latest block within the canonical chain before the start of slot $t-1$. In this attack, the adversary does not publish block $B^\mathcal{A}_{t-1}$ at slot $t-1$ and delays publishing it to the next slot. At the beginning of slot $t$, when the leader of slot $t$ proposes block $B_t$ on top of block $B_0$, the attacker simultaneously publishes its block $B^\mathcal{A}_{t-1}$ proposed on top of $B_0$. The adversary sets the game rule for the attestors of slot $t$ as follows: "The attestors of slot $t$ should vote for the block $B^\mathcal{A}_{t-1}$. Otherwise, their votes will not be included in the block of slot $t+1$". At the start of slot $t+1$, adversary $\mathcal{A}$ takes the following action:
\begin{itemize}[leftmargin=*]
    \item If over half of slot $t$ attestors vote for $B_{t}$, it proposes its block $B^\mathcal{A}_{t+1}$ on top of block $B_{t}$ and refrains from including the non-compliant votes. In this scenario, $\mathcal{A}$ is unsuccessful.
    \item If over half of slot $t$ attestors vote for $B^\mathcal{A}_{t-1}$, it proposes its block $B^\mathcal{A}_{t+1}$ on top of block $B^\mathcal{A}_{t-1}$ and includes the compliant votes. Then, $\mathcal{A}$ is successful.
\end{itemize}
Using the same analysis as presented in Section~\ref{sec:simple-game-analysis}, one can show as long as more than half of slot $t$ attestors are rational, the adversary can succeed in the simple attack without proposer boost. 
If $\Wp>0$, this attack can succeed even with fewer than 50\% rational attestors.

\end{document}